\documentclass{article}
\usepackage[margin=1in]{geometry}
\usepackage{graphicx}
\usepackage[utf8]{inputenc}
\usepackage[toc]{appendix}
\usepackage{amsmath}
\usepackage{amsthm}
\usepackage{amssymb}
\usepackage{enumitem}

\newtheorem{thm}{Theorem}[section]
\newtheorem{lem}[thm]{Lemma}

\newtheorem{cor}{Corollary}

\theoremstyle{definition}

\setlist[itemize]{leftmargin=*}

\begin{document}

\title{Decentralized network congestion control for DAG-based distributed ledger system}                      
%\tnotemark[1,2]

%\tnotetext[1]{This document is the results of the research
%   project funded by the National Science Foundation.}

% \tnotetext[2]{The second title footnote which is a longer text matter
%    to fill through the whole text width and overflow into
%    another line in the footnotes area of the first page.}

\author{Mayank Pandey, Rachit Agarwal,
Sandeep Kumar Shukla, Nishchal Kumar Verma\\
		IIT Kanpur,	Kanpur, India \\
	\texttt{\{pandeym,rachitag,sandeeps,nishchal\}@iitk.ac.in}}

\date{}
\maketitle

\begin{abstract}
We propose a variable and behavior-based node-specific proof-of-work (PoW) model for a directed acyclic graph (DAG)-based distributed ledger technology (DLT) network to mitigate decentralized network congestion control. Network congestion control for centralized communication systems is an established field of study, with detailed and continuous research being done on the subject. However, attention to congestion control in decentralized networks is relatively recent and underexplored, especially with DLT, such as blockchain and DAG-based networks. For the DLT networks, the network congestion is caused by factors such as transaction spamming, an increase in the user base, and the launch of new tokens. We focus on the congestion caused by the spamming of transactions within the blockchain and DAG-based DLT network. Based on the network throughput of transactions per second and consensus procedure, the DAG-based DLT needs to control network spamming more than the blockchain network. The PoW model within the DLT consensus framework is a limited deterrent against spamming. Our model provides equal opportunities for all stakeholders regardless of their computational resources. It prevents and penalizes any node that attempts to spam or dominate the network with more than the prescribed number of transactions. Since the system nodes compete to issue transactions with finite network resources, we display the system behavior through a non-cooperative game. Further, we show that our model enforces prescribed behavior amongst the nodes through the proof of the existence of Nash equilibrium in the game. 
\end{abstract}

% \includegraphics{figs/cas-grabs.pdf}

% \begin{highlights}
% \item Research highlights item 1
% \item Research highlights item 2
% \item Research highlights item 3
% \end{highlights}

% \begin{keywords}
% Blockchain \sep Directed acyclic graph \sep Distributed ledger \sep Congestion control \sep Proof of Work \sep Spamming \sep Nash equilibrium
% \end{keywords}

\maketitle

\section{Introduction}

%Distributed Ledger Technologies (DLT), such as blockchain, are increasingly attracting attention as an alternative to centralized solutions. 
In this paper, we conceptualize and formulate a variable proof of work (PoW) model for a Directed Acyclic Graph (DAG)-based Distributed Ledger Technology (DLT) network to tackle the problem of decentralized network congestion control. Congestion control is a significant requirement for the seamless functioning of communication networks. For centralized systems, it is a well-studied area~\cite{low2002internet}. However, the congestion control for decentralized networks warrants a fresh and separate approach, especially after the advent of DLT networks. Few application-specific methods for the same exist, such as in vehicular network messaging~\cite{balador2021survey}. However, broader and more in-depth work is still required to develop DLT networks' congestion control methodologies.

The DLT networks, such as blockchain, provide a new method for peer-to-peer transactions and communication with improved privacy and security~\cite{dorri2017blockchain}. However, network congestion is still possible in such systems. For example, an increase in the DLT user base leads to an increase in transaction-generating points. Additionally, a smart contract-based DLT network faces pressure when a new token gets launched due to the addition of assets and subsequent increases in the number of transactions over the network. Both aforementioned scenarios are not illicit behavior as they are part of the DLT operations. Apart from these two, the spamming of transactions by a malicious user also leads to the unnecessary hold of communication resources, causing congestion in the network. We focus on preventing transaction spamming to control congestion and ensure the seamless functioning of the DLT network. As we proceed, we look at the specific DLT features related to network traffic~\cite{blocksurveyLi}.

The functioning of a blockchain is dependent on miners or high stakeholders for transaction inclusion and confirmation through forming and adding blocks. The consensus process based on either PoW~\cite{fullmer2018analysis} or proof of stake (PoS)~\cite{saleh2021blockchain} keeps essential communication to a bare minimum while verifying the transactions. The resource requirement for adding the blocks ensures that the miners or stakers are selective in terms of block formation and its broadcasting. At the individual level, the network traffic-dependent transaction fees discourage users from spamming the network and encourage them to post only relevant transactions. The miners or stakers use the transaction fees offered to decide whether to include the concerned transaction in the block. The PoW process is a limited deterrent against transaction spamming by restricting the broadcast of proposed block data. However, it also causes poor throughput compared to centralized systems. Another disadvantage is low scalability due to the competitive structure of the consensus process for including the block. It also leads to the isolation of low-resource users, such as those using IoT devices, and the discouragement of essential micro-transactions. Such shortcomings paved the way for the conception of DAG-based DLT networks. 

%These miners or stakers use the transaction fees offered to decide whether to include the concerned transaction in the block. It restricts users from sending unnecessary transactions for the purpose of spamming.
%However, it also causes poor throughput compared to centralized systems and low scalability due to the competitive structure of the consensus process for including the block. Such characteristics, while causing problems of their own, are the only measures to prevent transaction spamming in the blockchain network to some extent. 

%Aside from blockchain,
DAG based DLT offer a viable alternative to blockchain for decentralized transaction networks~\cite{ben2018dag}. They have almost instantaneous transaction confirmation~\cite{pervez2018comparative}, ensuring the participation of almost all participants~\cite{reyna2018blockchain}. In a DAG-based DLT, the nodes collectively increase the speed throughput by adding their transactions and validating existing transactions. They cooperate and confirm each others' transactions instead of competing to add their own set of transaction blocks like in blockchain~\cite{kotilevets2018}. Due to this, DAG-based networks are able to operate without the transaction fee requirement. Hence, throughput and scalability for the network are comparatively higher for DAG-based DLT~\cite{chainordag}. However, the same characteristics of DAG-based DLT, which lead to scalability, also make the network vulnerable to transaction spamming by unscrupulous users and, subsequently, cause network congestion. While blockchain networks have the measures of PoW and transaction fees, DAG-based DLT lacks their equivalent of the same. It increases the individual user's power in DAG-based DLT to add transactions. Therefore, we focus on network congestion in DAG-based DLT caused by transaction spamming. 
%In this paper, we propose a solution for the same. 

%A notable DAG-based ledger is the Tangle, which the IOTA foundation developed. The Tangle-based decentralized network offers a communication and transaction framework for IoT devices without the inclusion of transaction fees. It also shares the problems related to communication systems. One such problem is the occurrence of network congestion and spamming.  
% Shortcomings of the blockchain ledger and why DAG is better. Build up to why IOTA is better than other DAGs. Use SOK wala survey paper.

%DAG-based distributed ledgers have almost instantaneous transaction confirmation~\cite{pervez2018comparative}. They also ensure the participation of small devices such as in IoT~\cite{reyna2018blockchain}. For DAG based ledger, any node adding a transaction into the ledger has to verify the existing transactions and attach itself to them. The ledger runs on cooperation between the nodes verifying each others' transactions. There is no wastage of resources in block formation efforts or competition amongst the miners or stakers to get their proposed block included. Each user is equally responsible for adding and verifying transactions in the ledger. It also ensures that the system can run without the transaction fee requirement. The aforementioned mitigation of blockchain shortcomings increases the individual user's power in DAG-based DLT to add transactions. It makes the network vulnerable to congestion due to transaction spamming. Therefore, we discuss the approach to resolve the same.    

In DAG-based ledgers, a nominal PoW is required as a "proof of verification and attachment"~\cite{coordicide}. It differs from the blockchain's PoW, which is the basis for competition between miners to add the block. DAG-based systems that use PoW based attachment include IOTA~\cite{silvano2020iota}, Graphchain~\cite{graphchain}, Phantom~\cite{2018phantom}, and Meshcash~\cite{meshcash}. Originally, the objective of the nominal PoW does not include preventing transaction spamming. Therefore, theoretically, any node can attach any number of transactions to the DAG ledger. In real-time, it causes many issues. Firstly, the communication resources in the form of channel and network services are limited. A node with higher computational resources at its disposal can issue transactions at a much faster rate and spam the network, thereby restricting other users.
Additionally, a group of nodes with significant resources can collude amongst themselves to get their invalid transactions, such as double-spending, verified by spamming the ledger. There is an upper limit on the available communication resources in terms of the number of transactions due to the periodic requirement of syncing the ledger. Above this limit, the network performance gets affected with users not being able to broadcast any additional transactions without delay~\cite{YANG2006175}. It affects their inclusion in the ledger and sometimes invalidates the transaction itself due to the expiry of its validity~\cite{coordicide}. 

Among the DAG-based DLT, the difficulty level of IOTA's PoW process is variable and unique to the behavior of each participating node~\cite{coordicide}. It is formulated to prevent transaction spamming only. However, its methodology is not sufficient to completely deter the users from spamming the IOTA ledger (also known as the Tangle). Besides IOTA's PoW process for transaction attachment, there is no effective method to deter transaction spamming in DAG-based DLT networks. After observing the existing shortcomings, we proceed to provide a solution to prevent such spamming. To facilitate this, we propose a user reputation-based difficulty level model for the DAG-based DLT networks, which is also applicable to other DLT networks with modifications. We provide a detailed description of our methodology and prove its efficacy through the numerical simulations of the functioning based on a non-cooperative game. Through our proposed methodology, We achieve effective congestion control for DAG-based DLTs without compromising the overall scalability.

\subsubsection*{Our contributions}
%Based on the observations, 
Our core contributions are listed below.
\begin{itemize}
    \item We propose a \textbf{\textit{congestion control model}} for a DAG-based DLT network preventing transaction spamming by the nodes. It prescribes and establishes the ideal user behavior in terms of carrying out the transactions. Our method provides every user with an equal opportunity to participate and imposes a penalty for deviating from the laid-down line of action.
    \item We establish the \textbf{\textit{efficacy of the prescribed user behavior}} through our model by the establishment of Nash equilibrium in a non-cooperative game with respect to the addition of transactions to the ledger.
    \item \textbf{\textit{Continuity of prescribed behavior}} in the event of a change in the network user composition is established through a uniform change in the Nash equilibrium.
\end{itemize}     

The remaining part of our manuscript is organized as follows. Section~\ref{sec:prelim} provides the existing shortcomings in decentralized congestion control with the motivation for our proposed method. Section~\ref{sec:model} describes the congestion control model for the decentralized system with a distributed DAG-based ledger. Section~\ref{sec:nash} shows the establishment of Nash equilibrium for cooperative behavior for stationary as well as dynamic sets of nodes. Section~\ref{sec:example} demonstrates our model's efficacy through numerical examples and the simulation. Finally, section~\ref{sec:conc} concludes our findings with a brief discussion. 

\section{Background and Existing shortcomings}
\label{sec:prelim}
%The system in question contains two graphs, $\mathbb{G}_{N}$ and $\mathbb{G}_{L}$.
In this section, we discuss the congestion control situation for DLT-based decentralized networks. For blockchain networks, the competition-based PoW, like in Bitcoin, acts as a deterrent against the spamming of proposed blocks. Similarly, the PoS consensus mechanism discourages spamming by enforcing block formation probability to be directly proportional to the cryptocurrency staked by the validator. In a traditional blockchain, the addition of a transaction to the ledger happens in stages. Firstly, the transaction goes into the pool of pending transactions. Then, from the pool, it is picked up by the miner or validator and added to the new block proposed for adding to the blockchain. The block reward and transaction fees on individual transactions act as incentives to carry out the whole process.  
%The users carry out the double-stage process of creating and adding the block as a competition in lieu of the transaction fees for each transaction. 
The aforementioned methods and characteristics are either resource-consuming or high-stakes favoring. Nonetheless, they ensure the reluctance of blockchain users to unnecessarily clog the DLT network through spamming of transactions. 
%After blockchain, we look into the DAG-based DLT networks 

In the DAG-based distributed ledger, the transactions either get added directly by the user (blockless) or after accumulation into the blocks. Instead of being a linear sequence like in blockchain, the ledger structure in DAG-based DLT is a dynamic, expanding, unidirectional graph without any loops. For block-based DAG, the decentralized network incorporates blocks instead of individual transactions in a DAG-based structure. The protocols for mining blocks are PoW-based. In Phantom~\cite{2018phantom}, PoW-based mining for blocks is followed by a recursive k-clustering algorithm for their selection, which is competitive instead of cooperative. In Meshcash~\cite{meshcash} and Spectre~\cite{spectre}, nodes use PoW to create blocks in multiple rounds. Effectively, all block-based DAG DLT networks use PoW as proof of attachment to prevent spamming. Also, they create different types of messages for transactions, voters, confirmation, and proposer information. For the consensus, the existing blockchain methods of transaction fees and the two-stage process of transaction pool and block formation apply to block-based DAGs. In the context of preventing transaction spamming, they are at par with the blockchain networks. 

Due to similarities with the blockchain process, block-based DAGs favor the users with high computational resources at their disposal. The normal nodes, such as low-resource users and IoT devices, are disadvantaged in such cases. Such disadvantages are being taken care of in blockless DAGs. 
%Therefore, they are below blockless DAGs in terms of preference and adoption as an alternative to the blockchain. 

Among the blockless DAG-based networks, IOTA~\cite{coordicide} is currently the most popular with a generalized DAG-based ledger structure. There are other variants of IOTA, such as G-IOTA~\cite{giota} and E-IOTA~\cite{eiota}, with variations in the transaction verification process and attachment limit. In Graphchain~\cite{graphchain}, in addition to PoW, the transaction fees are left to collect for the users who attach their transactions to it and verify. Another type of blockless DAG structure is where the individual nodes maintain their multiple parallel chains of transactions with inter-chain connections when required, such as Hashgraph~\cite{hashgraph}. Such a structure is a restricted DAG ledger that favors high-resource nodes in forming longer chains and getting their transactions confirmed faster than others. 

For all the aforementioned DAG-based DLT networks, PoW is used as a deterrent against transaction spamming for congestion control. It is either nominal or adaptive. The adaptive PoW method is more responsive to transaction frequency than other networks. As we proceed, we explain how PoW works in the context of preventing transaction spamming. The congestion control model of IOTA in the form of adaptive PoW difficulty level ($d_{i}(t)$) is given in equation~\eqref{iotapow}, sourced from~\cite{coordicide}. All the adaptive PoW methods are more or less similar in terms of formula.    

\begin{equation}\label{iotapow}
    d_{i}(t) = d_{0} + \lfloor \gamma_{i} \times a_{i}(t) \rfloor
\end{equation}
Here, $d_{i}(t)$ is assigned for the user $\mathcal{V}^{N}_{i} \in \mathcal{V}^{N}$ at time $t$, where $\mathcal{V}^{N}$ is the set of users in the concerned DLT network. The term $d_{0}$ is the base difficulty level, which is the same for all users across the network. $a_{i}(t)$ represents the number of transactions $\mathcal{V}_{N,i}$ performed in the time period $[t-m,t]$ for $m\in \mathbb{N}$, while the term $\gamma_{i} \in [0,1]$ is associated with its reputation known as mana. Before returning to the role and significance of the aforementioned terms, we first explain the role of $d_{i}(t)$ in the PoW-based process. 

In the PoW method, the user applies brute force technique to compute an output value within a pre-specified range through hash function~\cite{hashfunc}, where input is the transaction or block data as per the consensus requirements. The average computation time required for the same is proportional to the difficulty level. Suppose the standard output value has $D$ number of binary digits and the difficulty level is $d_{i}(t)$. It is used to set specific output targets, such as finding hash output with a value less than $2^{D-d_{i}(t)}$. The one given in equation~\eqref{iotapow} is IOTA's PoW which is adaptive to the user reputation. As brute force is the only viable technique in this case, the average computational work required is proportional to $2^{d_{i}(t)}$. Therefore, if $d_{i}(t)$ increases by $1$, then the requirement gets doubled. In case of a decrease by $1$, it gets halved. Hence, the average computational resource requirement $\Upsilon_{i}(t)$ is proportional to assigned difficulty level $d_{i}(t)$ for the respective user $\mathcal{V}^{N}_{i}$ as shown in~\eqref{resconsup}. 

\begin{equation}\label{resconsup}
    \Upsilon_{i}(t) \propto 2^{d_{i}(t)}
\end{equation}
For the nominal PoW, we have $d_{i}(t)=d_{0}$, i.e., the difficulty level remains the same across the time for the concerned DLT network. The variable difficulty level, $d_{i}(t)$, changes with time based on the parameters defined in the formula.

Recalling from the equation~\eqref{iotapow}, $\gamma_{i} \in [0,1]$ is associated with the term ``mana'', used to represent reputation for IOTA users. The more mana a user has, the lesser the value of $\gamma_{i}$ will be. The model for the pending mana and mana from~\cite{coordicide} is given in~\eqref{pendmana}.
\begin{equation}\label{pendmana}
\begin{aligned}
   & m_{i}(t) = m_{i}(0)e^{-\gamma t} + \frac{\alpha S_{i}}{\gamma}\left(1 - e^{-\gamma t}\right)\\
   & M_{i}(t) = M_{i}(0) e^{-\gamma t}
\end{aligned}
\end{equation}
In eq~\eqref{pendmana}, the term $m_{i}(t)$ is the pending mana for node $i$ at time $t$, while $M_{i}(t)$ is the mana. $S_{i}$ is the amount of token node $\mathcal{V}_{N,i}$ holds at time $t$, while $\alpha$ and $\gamma$ are the generation and decay rate respectively. $m_{i}(t)$ is converted to $M_{i}(t)$ when node $i$ spends token. The reduction in $m_{i}(t)$ due to reduction in $S_{i}$ at time $t$ gets added to $M_{i}(t)$. The pending mana both generates and decays at a pre-decided rate, while the mana only decays. The mana $M_{i}(t)$ is used to determine the reputation $\gamma_{i}$ for node $i$ at time $t$ in~\eqref{iotapow} for the PoW difficulty level $d_{i}(t)$. 

The equation~\eqref{iotapow} can also be used as a generalized model explaining different types of PoW methods. For instance, in Bitcoin, the variables in the term $\lfloor \gamma_{i} \times a_{i}(t) \rfloor$ can represent the average time period between the addition of two successive blocks, based on PoW model given in~\cite{bitcoinwhite}. The higher-than-intended time gap leads to a decrease in $d_{i}(t)$, while the lower time gap increases it. 

The state-of-the-art PoW methods, as per the equation~\eqref{iotapow}, favor the users having high resources based on the consensus process, whether the computational resource or cryptocurrency stake in the concerned network. The more the throughput in terms of transactions per second (TPS) is, the more the concerned DLT network is vulnerable to transaction spamming. Our objective is to create a congestion control method through the difficulty level, which controls spamming without either the support of variable transaction fees or any compromise on the network throughput capacity. Also, we aim to provide the methodology that is viable for all the PoW-based DLT networks, whether nominal or competitive. As we proceed, we describe in detail our proposed methodology followed by its validation in terms of fairness of transaction issuing opportunity.

\section{Proposed congestion control method}
\label{sec:model}
To present the model, let us assume a dynamic decentralized system operating a distributed DAG-based DLT network represented as $\mathbb{G}(T)$. For simplicity, we represent $\mathbb{G}(T)$ as $\mathbb{G}$ with all its parameters being described for the time instance $T$. The system $\mathbb{G}$ is a combination of two dynamic graphs, $\mathbb{G}^{N}$ and $\mathbb{G}^{L}$. 
Here, $\mathbb{G}^{N}$ is the graph of connected participating nodes, while $\mathbb{G}^{L}$ is the graph of the distributed ledger. $\mathbb{G}^{N}$ is a bidirectional graph represented as $\mathbb{G}^{N} = (\mathcal{V}^{N}, \mathcal{E}^{N}, \mathcal{A}^{N})$. Here, $\mathcal{V}^{N}$ represents the decentralized network users as a set of nodes given by $\mathcal{V}^{N} = \{\mathcal{V}^{N}_{1}, \mathcal{V}^{N}_{2}, \cdots, \mathcal{V}^{N}_{n}\}$, numbered in the order of their joining the system. Let the cardinality of the set $\mathcal{V}^{N}$ be given as $n$. It is a variable term as the users are free to join and leave the concerned DLT network. The expression $\mathcal{E}^{N} \subset \mathcal{V}^{N} \times \mathcal{V}^{N}$ represents the set of edges representing the communication channel between nodes. Further, $\mathcal{A}^{N}$ is the $n\times n$ adjacency matrix depicting the information on connecting unidirectional edges, with the elements of matrix $a^{N}_{i,j} \in \{0,1\}$ depicting information of connection between nodes.  

On the other hand, $\mathbb{G}^{L}$ has a DAG-shaped structure. It represents the distributed ledger of transactions defined as $\mathbb{G}^{L} = (\mathcal{V}^{L}, \mathcal{E}^{L}, \mathcal{A}^{L})$. $\mathbb{G}^{L}$ is dynamic and expands rapidly with time as new transactions are made. Here the set of nodes $\mathcal{V}^{L}$ represents the transactions, given by $\mathcal{V}^{L} = \{\mathcal{V}^{L}_{1}, \mathcal{V}^{L}_{2}, \cdots, \mathcal{V}^{L}_{l}\}$, numbered in the order of their arrival or addition to the graph. Let the cardinality of the set $\mathcal{V}^{L}$ be given as $l$, which is incremental as the new transactions continuously get added to the concerned ledger. The set of edges $\mathcal{E}^{L} \subset \mathcal{V}^{L} \times \mathcal{V}^{L}$ are unidirectional and always directed from the nodes added later towards the nodes added earlier. The $l\times l$ adjacency matrix $\mathcal{A}^{L} $ stores the information on transaction attachment such that the elements of matrix $a^{L}_{i,j} \in \{0,1\}$. 

The elements in $\mathcal{V}^{N}$ create, issue, and add the transactions into the ledger $\mathbb{G}^{L}$. The users in $\mathcal{V}^{N}$ are connected in a decentralized network, and each user has the information on the state of $\mathbb{G}^{L}$. The state of $\mathbb{G}^{L}$ is updated continuously by the addition of new transactions by users in $\mathcal{V}^{N}$. Every transaction in $\mathcal{V}^{L}$ is connected to at least two previously added transactions though edges belonging to $\mathcal{E}^{L}$. A depiction of a DAG-based DLT network is given in the figure~\ref{dagfig}. In $\mathbb{G}^{L}$, the transactions at the end which are not yet verified are referred to as "tips". Therefore, the procedure of selecting the existing transactions for verification and attachment is called the tip selection process.

\begin{figure}
\includegraphics[width=9.25cm]{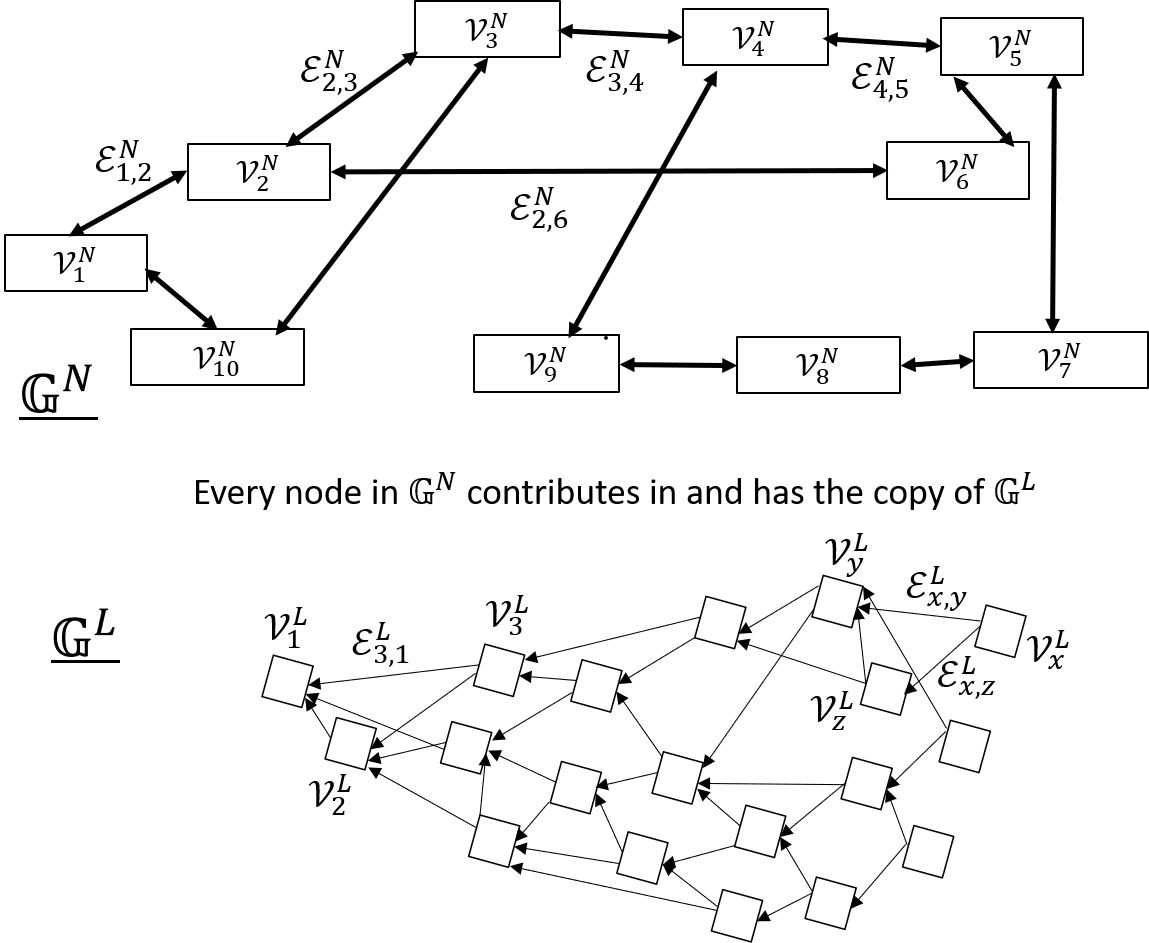}
\centering
\caption{Representation of $\mathbb{G}^{N}$ and $\mathbb{G}^{L}$ }
\label{dagfig}
\end{figure}

Before proceeding further, we would like to declare that we will use the terms $\mathcal{V}^{N}_{i}$ and $i$ interchangeably, based on the requirement in representation. The functioning of our proposed method is explained through the conditions regarding the state of $\mathbb{G}$ given below.
\begin{enumerate}
    \item The parameters $n$ and $l$ are time-variant terms, with the rate of change of $l$ much higher than that of $n$.
    \item Initially, we proceed with the assumption that $\frac{dn}{dt} = 0$, i.e., the number of users in the DLT system does not change.
    \item Afterwards, we also test our method with $\frac{dn}{dt} \neq 0$ with the condition of existing users leaving the network and new users joining.
    \item  Each element of $\mathcal{V}^{L}$ (transactions) stores some information, such as involved user addresses, amounts, and metadata.
    \item  Each element of $\mathcal{E}^{L}$ (connecting edges between transactions) has unity weight. 
    \item A node in $\mathcal{V}^{N}$ adds at least a predecided number of edges ($2$ in our case) in $\mathcal{E}^{L}$ while adding a single node in $\mathcal{V}^{L}$.
\end{enumerate}
In the DLT network, the nodes in $\mathcal{V}^{N}$ can add the transactions into $\mathbb{G}^{L}$ up to a given limit during a particular time period due to real-time limits on communication resources~\cite{JIN20141}. This limit is represented through transactions per second (TPS), which is the average capacity of the concerned DLT network. Let the limit be represented as $\Lambda$ number of transactions for a given time period of $T$ seconds, with the average capacity being $\frac{\Lambda}{T}$ TPS. 

%$\mathbb{G}_{N}$ is dynamic, meaning a node in $\mathcal{V}_{N}$ can either join and leave at any time or have intermittent connectivity with other nodes in $\mathbb{G}_{N}$. The $\mathcal{E}_{N}\in \mathbb{G}_{N}$ reflects the information on the interconnection between the nodes. Further, note that $ [a_{N,i,j}] \in \{0,1\} \forall i,j \in [1,n]$, where 0 means that there is no connection between the nodes $i$ and $j$ in $\mathcal{V}_{N}$ whereas 1 means that there is a connection between them. Each node in $\mathbb{G}_{N}$ maintains a full copy of $\mathbb{G}_{L}$. 

%For the decentralized network of $\mathcal{V}_{N} \in \mathbb{G}_{N}$, every user in $\mathcal{V}_{N}$ can join and leave at any time. The $\mathcal{E}_{N}\in \mathbb{G}_{N}$ stores the information on the interconnection within the nodes. It is not necessary that every node is directly connected to every node in $\mathcal{V}_{N}$, i.e, $ [a_{N,i,j}] \in \{0,1\} \forall i,j \in [1,n]$. The nodes communicate across the network through gossip. Since the network operates a distributed ledger, each user in $\mathcal{V}_{N}$ has access to a full copy of $\mathbb{G}_{L}$.  Each user is free to add a transaction at any instant without being overseen by an administrative authority. Hence,  

Each user in $\mathcal{V}^{N}$, after creating a transaction, selects a pre-decided number of the existing transactions in the $\mathbb{G}^{L}$ to verify and attach its transaction to them. These existing transactions are referred to as tips. We keep the required number of tips at $2$ as it is an acceptable number in most existing DAG-based DLT networks. After selection, the PoW is created according to the assigned difficulty level and attaches its transaction to the selected tips after their verification. Then it broadcasts the complete transaction state. There is a realistic possibility that a single or a group of malicious nodes with high computational resources would try to dominate the ledger by adding only their own transactions into $\mathbb{G}_{L}$ by quickly solving the nominal PoW. The distributed ledger-based decentralized networks design their consensus protocol to mitigate such phenomenon.

Let $T_{1}, T_{2}, \cdots T_{\infty}$ be the successive discrete time instances with non-uniform periods between them. The structure of a transaction in $\mathcal{V}^{L}$ is given in~eq\eqref{transtruc}, where $T_{y},T_{z} < T_{x}$.
\begin{equation}\label{transtruc}
    \mathcal{V}^{L}_{x}(T_{x}) = \{ \mathcal{V}^{L}_{y}(T_{y}), \mathcal{V}^{L}_{z}(T_{z}), D^{L}_{x}, \mathcal{E}^{L}_{x,y}, \mathcal{E}^{L}_{x,z}, \nu_{x}\}
\end{equation}
Suppose a user $\mathcal{V}^{N}_{i}$ attaches a transaction $\mathcal{V}^{L}_{x}(T_{x})$ into $\mathbb{G}^{L}$, with $T_{x}$. The term $D^{L}_{x}$ contains the data of the transaction $\mathcal{V}^{L}_{x}$ which contains the information such as which user issued the transaction, what is the value transferred and other system-specific information. The transactions $\mathcal{V}^{L}_{y}(T_{y})$ and $\mathcal{V}^{L}_{z}(T_{z})$ are the existing transactions in $\mathbb{G}^{L}$ to which $\mathcal{V}^{L}_{x}(T_{x})$ is attached through the edges $\mathcal{E}^{L}_{x,y}$ and $\mathcal{E}^{L}_{x,z}$ respectively. The term $\nu_{x}$ is the nonce value of transaction $\mathcal{V}^{L}_{x}(T_{x})$. It is computed and used in a similar way to the Bitcoin blockchain network. The nonce value $\nu_{x}$ is the proof of work (PoW) and the proof of attachment of a transaction into the DAG ledger.

%To understand how the PoW model works for the DAG-based distributed ledger, suppose the user $\mathcal{V}_{N,i}$ creates a transaction $\mathcal{V}_{L,x}(t_{x})$ and attaches it at instance $t_{x}$ to two existing transactions $\mathcal{V}_{L,y}(t_{y})$ and $\mathcal{V}_{L,z}(t_{z})$ with PoW difficulty level $d_{i}(t_{x})$. It then passes the information to the users $\mathcal{V}_{N,j} \in \mathbb{G}_{N} \forall j\neq i, [a_{N,i,j}]\neq 0 $. The users check the correctness of $\mathcal{V}_{L,x}(t_{x})$ and whether the appropriate difficulty level ($d_{i}(t_{x})$ in this case) was followed or not. 

%The only way to solve PoW with difficulty level $d_{i}(t_{x})$ is to apply brute force technique. The higher the value of $d_{i}(t)$ is, the more average time and resources it will take for $\mathcal{V}_{N,i}$ to solve the PoW associated with it. Therefore, the node with high and specialized computational resources has a higher probability of solving the PoW earlier than the node with the normal machine. 

% The difficulty level in equation~\ref{iotapow} is a positive integer value.  The base level $d_{0}$ is fixed for everyone in the network. The variable component must be adjusted accordingly to establish equality as well as to penalize for spamming. One way to implement is to incur a heavy cost on the violators without external intervention. The component $\gamma_{i}$ is dependent on the mana, which depends on the spending of the user $\mathcal{V}_{N,i}$. The component $a_{i}(t)$ is positive as long as $\mathcal{V}_{N,i}$ issues a transaction message.  

We aim to incur the minimum possible cost in terms of computational and communication resources with the penalty for crossing the limit. Therefore, we propose the user reputation-based difficulty level model for the decentralized network with a distributed DAG-based ledger. 

We present our methodology in a discrete-time frame with a uniform period. Let the time instances be represented as $\{0,T,2T,\cdots,(C-1)T,CT,\cdots\}$ with time period of $T$, where $C \in \mathbb{N}$. Our method defines the user-specific difficulty level for the PoW required for transaction attachment based on the concerned user's reputation. We define the reputation alloted for the user $\mathcal{V}^{N}_{i}$ for the time period $[CT, (C+1)T]$ as $R^{N}_{i}(CT)$. It is generated based on the spending of tokens between the duration $[(C-1)T, CT]$. $R^{N}_{i}(CT)$ has two components; one is related to the value of spent tokens through their transactions, and the other is related to the number of transactions issued, both during the time period $T$. Here tokens refer to the cryptocurrency of the concerned DLT network. The values mentioned above are computed at regular intervals and remain fixed for a time period of $T$. We assume that the time synchronization across the nodes is uniform, i.e., Universal coordinated time~\cite{panfilo2019}. The overall reputation model is given in~\eqref{fullrepmodel}, with the two components $f^{B}_{i}(CT)$ and $f^{w}_{i}(w_{i}(CT))$ further defined. 
\begin{equation}\label{fullrepmodel}
    R^{N}_{i}(CT) = F(f^{B}_{i}(CT),f^{w}_{i}(w_{i}(CT)))
\end{equation}
The component $f^{B}_{i}(CT)$ has the range of $[0,1]$, and it represents the relative value in terms of tokens transferred by $\mathcal{V}^{N}_{i}$ between the instance $[(C-1)T, CT]$, with respect to the total number of tokens transferred during that period. If $\mathcal{V}^{N}_{i}$ does not transfer any tokens during $[(C-1)T,CT]$, then $f^{B}_{i}(CT) = 1$. On the other hand, if $\mathcal{V}^{N}_{i}$ is responsible for all the token transfers in the network during $[(C-1)T, CT]$, then $f^{B}_{i}(CT) = 0$. The model for formulating the value of $f^{B}_{i}(CT)$ is given in~\eqref{balrep}. Here, $B^{N}_{i}(CT)$ is the token balance of $\mathcal{V}^{N}_{i}$ at instance $CT$, while $\Delta B^{N}(CT)$ is the total token amount in the ``sent'' for all the transactions in the network during the duration $[(C-1)T,CT]$. The objective for $f^{B}_{i}(CT)$ is to reward the spending, but up to a certain limit.  
\begin{equation}\label{balrep}
% \begin{split}
%     f_{B,i}(CT) & = 1  \;\;\;\;\;\;\;\;\;\;\;\;\;\;\;\;\;\;\;\;\;\;\;\;\;\;\;\;\;\;\;\;\;\;\;\;\;\;\;\;\; \Delta B_{N}(CT) = 0\\
%     & = 1 - \frac{B_{N,i}((C-1)T) - B_{N,i}(CT)}{\Delta B_{N}(CT)} \;\; \Delta B_{N}(CT) > 0
% \end{split}
f^{B}_{i}(CT)=\left\{\begin{matrix}
1 & \text{if}  \Delta B^{N}(CT) = 0 \\ 
1 - \frac{\max \{0, B^{N}_{i}((C-1)T) - B^{N}_{i}(CT)\}}{\Delta B^{N}(CT)} &   \text{if} \Delta B^{N}(CT) > 0
\end{matrix}\right.
\end{equation}
To encourage active participation, but discourage spamming beyond a certain point, the suitable function for the reputation $f^{w}_{i}$ must be diminishing in structure represented via~\eqref{dimfunc}. Here, $w_{i}(CT)$ is the number of transactions added by $\mathcal{V}^{N}_{i}$ to $\mathbb{G}^{L}$ during the period $[(C-1)T,CT]$. The network capacity is divided among the users as in~\eqref{netcap}. 
\begin{equation}\label{netcap}
    \Lambda = n K_{cross} + K_{add}
\end{equation}
Here, $K_{cross} \in \mathbb{N}$ is defined as the average transaction capacity of each user for $n$ users in the system. $K_{add} << K_{cross}$ is the residual system capacity. Based on $K_{cross}$, we choose the value of the terms $K_{0}$, $K_{1}$ and $K_{2}$ to form a quadratic equation. The values are chosen such that one solution is $K_{cross}$, while the other is a random negative integer. The negative value gets discarded as the number of transactions cannot be negative. The transaction frequency-based reputation equation and $K_{cross}$ as the only viable solution is given in~\eqref{dimfunc}. 
\begin{equation}\label{dimfunc}
\begin{aligned}
    & f^{w}_{i}(w_{i}(CT)) = K_{0} + K_{1}w_{i}(CT) - K_{2}w_{i}^{2}(CT) \\ 
    & K_{cross} = \frac{K_{1}+\sqrt{K_{1}^{2}+4K_{2}K_{0}}}{2K_{2}}  
\end{aligned}
\end{equation}
The total number of transactions added to $\mathcal{V}^{L} \in \mathbb{G}^{L}$ between $[(C-1)T,CT]$ is represented as $W(CT) = \mathcal{V}^{L}(CT) - \mathcal{V}^{L}((C-1)T) = \sum_{i \in \mathcal{V}^{N}} w_{i}(CT)$. In order for the decentralized network to work properly, the condition given is $W(CT) \leq \Lambda$. In theoretical terms, each node can unilaterally use the whole network resources for itself, i.e., $ w_{i}(CT)^{max} =  \Lambda$. 
%At $w_{i}(CT) = 0$, the value of $f_{w,i}(CT)$ will be $K_{0}$. For the given parameters, the term $f_{w,i}(CT)$ will first increase and then decrease. In theoretical terms, each node can unilaterally use the whole network resources for itself, i.e., $ w_{i}(CT)^{max} =  \Lambda$.

We formulate the PoW difficulty level based on the reputation components $f^{B}_{i}(CT)$ and $f^{w}_{i}(w_{i}(CT))$ to prevent spamming and network congestion. The proposed difficulty level equation for the node $\mathcal{V}^{N}_{i}$ for the period $[CT, (C+1)T]$ based on the transaction frequency in the period $[(C-1)T, CT]$ is given in~\eqref{newdifflevel}.
\begin{equation}\label{newdifflevel}
\begin{aligned}
    d_{i}(CT) = &  d_{0} - \left\lfloor \frac{df^{w}_{i}(w_{i}(CT))}{dw_{i}(CT)}  \right\rfloor \\ 
    & + \lceil f^{B}_{i}(CT) \times (- \min\{0,f^{w}_{i}(w_{i}(CT)))\} \rceil
    \end{aligned}
\end{equation}
% As explained previously, the linear increase in the difficulty level leads to two times the requirement for more resources to add the same transactions. From equation~\ref{intvalue}, we find the maximum number of transactions $\mathcal{V}_{N, i}$ can carry out between $[(C-1)T, CT]$ with minimum computational work. 
The term $d_{0}$ is the base difficulty level. The derivative component is used for linear increment till $w_{i}(CT) \leq K_{cross}$, while the third term is for quadratic increment when $w_{i}(CT) > K_{cross}$. For the transaction frequency $w_{i}(CT)$ in the period $[(C-1)T, CT]$, the node $\mathcal{V}_{N,i}$ has to add transactions with PoW difficulty level $d_{i}(CT)$ for the entire period of $[CT, (C+1)T]$. From~\eqref{newdifflevel}, the objective of $\mathcal{V}^{N}_{i}$ is to utilize minimum individual computational resources. The node has to find the balance between the transaction frequency $w_{i}(CT)$ and the subsequently allotted difficulty level $d_{i}(CT)$. Meanwhile, the total number of transaction limit in a period $T$ are restricted to $\Lambda$ due to physical infrastructure and synchronization requirement of the network. The scenario is ideal to be represented as a non-cooperative game with a Nash equilibrium. Therefore, first we explain the system with the cases of the stationary and dynamic number of nodes and then we find the optimal solution using the designed non-cooperative game.  
 % \begin{mini}|l|
	%   {[(C-1)T,CT] }{\frac{d(d_{i}(CT)-d_{0})}{d w_{i}(CT)}\;\;}{}{}
 %      \addConstraint{\max w_{i}(CT)}{}
	%   \addConstraint{\Lambda - K_{add} \leq \sum_{i \in \mathcal{V}_{N}} w_{i}(CT)  \leq \Lambda}{}
	%   \addConstraint{\Delta B^{N}(CT) \geq 0}{}{}
 %   \label{singopt}  \end{mini}
% Based on~\eqref{singopt}, we represent the objective for the whole set of nodes $\mathcal{V}_{N}$ as given in~\eqref{redopt}. 
% \begin{mini}|l|
% 	  {[(C-1)T,CT] }{\sum_{i \in \mathcal{V}_{N}} \frac{d_{i}(CT)}{w_{i}(CT)}\;\;}{}{}
% 	  \addConstraint{\Lambda - K_{add} \leq \sum_{i \in \mathcal{V}_{N}} w_{i}(CT)  \leq \Lambda}{}
% 	  \addConstraint{\Delta B_{N}(CT) \geq 0}{}{}
%     \label{redopt} \end{mini}

Both $f^{B}_{i}(CT)$ and $(- min(0,f^{w}_{i}(w_{i}(CT))))$ are $\geq 0$. From~\eqref{balrep}, the condition for $f^{B}_{i}(CT)$ to become zero is given as $B^{N}_{i}((C-1)T) - B^{N}_{i}(CT) = \Delta B^{N}(CT)$. The condition is highly unlikely as it means only user $\mathcal{V}^{N}_{i}$ sent the tokens during the period $[(C-1)T,CT]$. This is highly unlikely to happen in a large decentralized network with multiple nodes. Therefore, we move to the term $(- min(0,f^{w}_{i}(w_{i}(CT))))$ for its minimizing conditions. We solve for the term $- \lfloor \frac{d f^{w}_{i}(w_{i}(CT))}{d w_{t}(CT)} \rfloor$. From~\eqref{dimfunc}, the expression $\frac{d f^{w}_{i}(w_{i}(CT))}{d w_{t}(CT)}$ becomes as $ \frac{d f^{w}_{i}(w_{i}(CT))}{d w_{t}(CT)} = K_{1} - 2K_{2}w_{i}(CT)$. The difficulty level $d_{i}(CT)$ values for varying ranges of $w_{i}(CT)$ is listed in appendix~\ref{app:list}. Our objective is to find the minimum value for the term $\frac{d_{i}(CT)}{w_{i}(CT)}$ for $0 \leq w_{i}(CT) \leq \Lambda$. On observing the values obtained for $d_{i}(CT)$ for the given range of $w_{i}(CT)$, the local minima occurs at $w_{i}(CT) = K_{cross}$. We prove this condition as the ideal scenario for the concerned DLT network through non-cooperative game formulation, described in the upcoming section.

% In a decentralized network, a node issues and adds a transaction to $\mathbb{G}_{L}$, then broadcasts it to other users in $\mathbb{G}_{N}$. Every node has the same level of permission to issue and add the transactions to the ledger graph. However, if any user sends too many transactions into the network, it may prevent others from transmitting their own due to finite communication resources. To counter the same, the concept of PoW-based attachment with variable difficulty levels is introduced. 

As explained above, the system assigns a difficulty level for the period $[CT, (C+1)T]$ after obtaining the transaction frequency input in the period $[(C-1)T, CT]$
. Based on the observations, it is in the interest of a rational node $\mathcal{V}^{N}_{i}$ to keep $ w_{i}(CT) \leq K_{cross}$ in order to optimize the combination of difficulty level vs. transaction frequency. From equation~\eqref{resconsup}, we know that the increase or decrease in the difficulty level by $1$ unit leads to the doubling or halving of the average resource requirement, respectively.

%Suppose a user $\mathcal{V}^{N}_{i}$ does $\leq \lfloor K_{tan}\rfloor$ transactions in the period $[(C-1)T,CT]$. In such case, the average resource consumption based on~\eqref{resconsup} for $\mathcal{V}^{N}_{i}$ during the period $[CT,(C+1)T]$ is represented as $\Upsilon^{\lfloor K_{tan}\rfloor}_{i}(CT) \propto  2^{d_{0}}$. 

For the total network capacity $\Lambda$, the optimal solution must be to utilize maximum capacity as well as provide a fair chance to each node. For $w_{i}(CT) = K_{cross}$, we have $d_{0} - \lfloor (K_{1} - 2K_{2} K_{cross})\rfloor = d_{cross}$. Then the average resource requirement is represented as $\Upsilon^{ K_{cross}}_{i}(CT) \propto 2^{d_{cross}}$. 

%For a node $\mathcal{V}^{N}_{i}$ with $\lfloor K_{tan}\rfloor w_{i}(CT) \leq K_{cross}$ transactions, the resource consumption is given as $ \Upsilon^{ K_{cross}}_{i}(CT) = M 2^{d_{0} - \lfloor (K_{1} - 2K_{2} K_{cross})\rfloor}$.

If $w_{i}(CT) = K_{cross} - 1$, then the subsequent resource requirement is represented as $\Upsilon^{ K_{cross} - 1}_{i}(CT) \propto 2^{d_{cross}- \lfloor 2K_{2}\rfloor}$. 

On the other hand, for $w_{i}(CT) = K_{cross} + 1$, the subsequent difficulty level for the period $[CT,(C+1)T]$ is given as $\Upsilon^{ K_{cross} + 1}_{i}(CT) \propto \left ( 2^{d_{cross}}\right)
     \times 2^{\lceil f^{B}_{i}(CT) \times (-K_{0} - K_{1}( K_{cross} + 1) + K_{2}( K_{cross} + 1)^{2} \rceil}$.
     
Based on the above observations, we represent the ratio of resource requirement for following and not following the prescribed behavior as given in equation~\eqref{eq:30}.
\begin{equation}\label{eq:30}
\begin{split}
     \Upsilon^{ K_{cross} + 1}_{i}(CT) \approx 2^{{ K^{2}_{cross}}} \times 
     \Upsilon^{ K_{cross}}_{i}(CT)
     \end{split}
\end{equation}
From equation~\eqref{eq:30}, it is evident that there is considerably heavier burden of the computational resources for the period $[CT,(C+1)T]$ if the $\mathcal{V}^{N}_{i}$ has $w_{i}(CT) >  K_{cross}$ by even $1$. It prompts a rational node to keep within the prescribed limit to avoid a disproportionate burden on its computational resources for successive time periods. As we proceed, we discuss the case of reputation with a dynamic set of nodes.
%Therefore, to utilize the full capacity of the network with optimal total resource usage, every user should issue $K_{cross}$ transactions.%As we proceed, we explore the ideal behavior in terms of issuance and transmission of transactions within the network. 
% For an individual user $\mathcal{V}_{N,i}$, there are $4$ scenarios for an individual user issuing transactions for the interval $[t-T,t], \forall t\geq T$. The same are given as below.
\subsection{Joining and leaving of nodes}
Until now, we explored the scenario where the number of nodes is fixed, i.e., $n$ is a constant value. We computed the PoW difficulty level allocation for the nodes for the period $[CT,(C+1)T]$ based on transactions done in the period $[(C-1)T,CT]$. 

%In $d_{i}(CT)$, the term $d_{0}$ is constant while the remaining part is variable and depends on $w_{i}(CT)$. During the period $[(C-1)T,CT]$, every user adds a certain number of transactions. For such a case, the individual difficulty level for the period $[CT,(C+1)T]$ will be decided based on the previous transaction history. 
For the situation where a new node $\mathcal{V}^{N}_{n+1}$ joins the network at a time instance falling in a particular time period, the reputation parameters get adjusted in the subsequent time period. For a generalized scenario, we assume that with the joining of the new node, the network capacity changes, i.e., either increases or decreases~\cite{JIN20141}. Suppose the new node joins during the time period $[(C- 1)T, CT]$. In such a case, our methodology will allow the new node to start issuing transactions from the time period $[CT, (C+1)T]$. For the variable node case, we represent the network capacity without the new user as $\Lambda^{(C-1)T}$ for the period $[(C- 1)T, CT]$. When the network capacity and the number of users change, the average network capacity also changes. The system defined values $K_{0}$, $K_{1}$ and $K_{2}$ update accordingly to the new capacity $\Lambda^{CT}$ starting from the period $[CT, (C+1)T$. Let the system defined parameter values be represented as $K^{(C-1)T}_{0}$, $K^{(C-1)T}_{1}$ and $K^{(C-1)T}_{2}$ for the period $[(C-1)T,CT]$. The transaction frequency-based reputation for the variable node set for the period $[(C- 1)T, CT]$ is represented via equation~\eqref{vardimfunc}. A similar form of change follows if a node leaves the network.
\begin{equation}\label{vardimfunc}
\begin{aligned}
& f^{w}_{i}(w_{i}(CT)) = K^{(C-1)T}_{0} + K^{(C-1)T}_{1}w_{i}(CT) \\
&- K^{(C-1)T}_{2}w_{i}^{2}(CT) 
\end{aligned}
\end{equation}
With the change in the number of users and network capacity, the average network capacity, i.e., $K_{cross}$ changes. For our purpose, let the number of users in the network before and after the change during the period $[(C-1)T, CT]$ be $n_{1}$ and $n_{2}$, respectively. In~\eqref{varprevnetcap}, we obtain the average network capacity for the period $[(C-1)T,CT]$, where $  K^{(C-1)T}_{cross}$ is the average capacity and $K^{(C-1)T}_{add}$ is the residual capacity.
\begin{equation}\label{varprevnetcap}
    \Lambda^{(C-1)T} = n_{1}  K^{(C-1)T}_{cross} + K^{(C-1)T}_{add}
\end{equation}

For the time period starting $[CT,(C+1)T]$ and onward, the average network capacity and, subsequently, reputation value parameters get updated according to $n_{2}$ users. The updated network capacity and the average capacity are represented in~\eqref{varupnetcap} with the variables having the same meaning as in~\eqref{varprevnetcap} for the updated time period. We provide a generalized scenario of the effect of change in user number coming from the period $[CT,(C+1)T]$. 

% The joining of new users either increases, decreases, or has no effect on the average network capacity depending on the individual user capacity. A high-resource user provides respite to the network by improving the average transaction capacity, while a low-resource user reduces the same. 

\begin{equation}\label{varupnetcap}
    \Lambda^{CT} = n_{2}  K^{CT}_{cross} + K^{CT}_{add}
\end{equation}
The difficulty level model for the fixed set given in~\eqref{newdifflevel} will operate in the same way for the variable node set but with varying network parameters.

Once we have the difficulty level model for Pow and adding transactions, we proceed to show that by following the proposed model and prescribed limits, the maximum possible transactions get added with the minimum possible total computational resource. We demonstrate the formulation of a non-cooperative game based on the variation of node reputation and prescribed limits in the difficulty level model. Subsequently, we establish the prominence of prescribed behavior through the establishment of Nash equilibrium in a non-cooperative game for both fixed and variable node sets.  
%Based on the aforementioned scenarios for the users issuing transactions, the resource usage and conditions for network operations can be derived as below. For the cases below, our objective is to optimise the computational resource consumption across the network at a time of $t$. The assumption for the below assessment is that, at instance, $t$, every user is issuing at least one transaction. 
\section{Optimal Resource consumption via non-cooperative games}
\label{sec:nash}

% For case I, where network capacity of $\Lambda \geq n \lfloor K_{cross}\rfloor$ is present, the network will remain underutilized if the individual nodes issue the transactions within their defined limits. For such scenarios, the network capacity can be utilized in two ways. The network should either provide entry to a new node or modify the individual user transaction limits accordingly. 
%In this section, we show the efficacy of the proposed model-based prescribed behavior by establishing a unique Nash equilibrium for the non-cooperative game between the nodes. 
The scenario in the previous section for $n$ users operating in a decentralized network is similar to a non-cooperative game. The reasons behind the similarity are that there is no communication between the users regarding their respective strategies, i.e., no user tells others how many transactions it is planning to issue and add during a time period. The physical resource constraints put a limit on the total number of transactions that can be added to the ledger within a time period without causing a backlog. Also, there is no binding agreement between the users to limit their transactions. We demonstrate the efficient utilization of the network through the prescribed behavior by showing the process in the form of a non-cooperative game. As we proceed, we show that the only way to maximize the possible utilization of network resources with minimum possible computational resources is at a uniform Nash equilibrium.  
\subsection{Nash equilibrium for a fixed set of nodes}
We define the non-cooperative game for the consideration in~\eqref{game}, whose payoff applies in the period $[CT,(C+1)T]$ based on the strategy carried out in the period $[(C-1)T,CT]$.
\begin{equation}\label{game}
    \Gamma  \triangleq \{\mathcal{V}^{N}, (w_{i}(CT))_{i \in \mathcal{V}^{N}}, (\Theta_{i})_{i \in \mathcal{V}^{N}} \}
\end{equation}
For the given game, we have the set of nodes $\mathcal{V}^{N} = \{\mathcal{V}^{N}_{1}, \mathcal{V}^{N}_{2}, \cdots, \mathcal{V}^{N}_{n}\}$ as a set of players. The system dynamics for the game are given in equations~\eqref{balrep}, \eqref{netcap}, \eqref{dimfunc}, and~\eqref{newdifflevel}. For $\mathcal{V}^{N}_{i} \in \mathcal{V}^{N}$, $w_{i}(CT)$ is the strategy of issuing the number of transactions in the period $[(C-1)T, CT]$. 

The payoff of the strategy applied in $[(C-1)T,CT]$ for $\mathcal{V}^{N}_{i}$ is $\Theta_{i}$, which gets allotted in the period $[CT,(C+1)T]$. The overall objective for the game in \eqref{game} is to utilize the full available network capacity in the current period as well as have optimal resource consumption for the subsequent period. The strategy vector for the system for the period $[(C-1)T,CT]$ is given by $w(CT) \triangleq [w_{1}(CT), w_{2}(CT), \cdots, w_{n}(CT)]^{T}$.

We design the payoff function function $\Theta_{i}$ for $\mathcal{V}^{N}_{i}$ by combining the components $U_{i}$ and $U_{av}$. For the component $U_{i}$ in~\eqref{util1}, the variable is $w_{i}(CT)$.
\begin{equation}\label{util1}
     U_{i} = \zeta_{1}w_{i}(CT) - \zeta_{2}w^{2}_{i}(CT)
\end{equation}
Also, the component based on the average strategy of the network $U_{av}$ is being given in~\eqref{avutil1}.
\begin{equation}\label{avutil1}
     U_{av} = \zeta_{1}\sum_{i \in \mathcal{V}^{N}}\frac{w_{i}(CT)}{n} - \zeta_{2}\sum_{i \in \mathcal{V}_{N}}\left(\frac{w_{i}(CT)+\cdots}{n}\right)^{2}
\end{equation}
The general form of the utility function $\Theta_{i} \in \Re$ is given below, where $\kappa_{1},\kappa_{2} > 0$ are weightage parameters. 
\begin{equation}\label{util}
    \Theta_{i} = \kappa_{1} U_{i} + \kappa_{2} U_{av}
\end{equation}
The above equation can be derived into the quadratic payoff form similar to given in~\cite{frihauf2011}, \cite{dokka2022}. We represent the payoff as given in~\eqref{derutil}.
\begin{equation}\label{derutil}
\begin{split}
    \Theta_{i} = \eta_{1} w_{i}(CT) + \eta_{2} w^{2}_{i}(CT) + \eta_{3} \sum^{j \neq i}_{j \in \mathcal{V}^{N}} w_{j}(CT)\\
    + \eta_{4} \sum^{j \neq i}_{j \in \mathcal{V}^{N}} w^{2}_{j}(CT) + \eta_{5}\sum^{j\neq i}_{i,j \in \mathcal{V}^{N}} w_{i}(CT) w_{j}(CT) 
    \end{split}
\end{equation}
The value of the constants are given as $\eta_{1} = \kappa_{1}\zeta_{1} + \frac{\kappa_{2}\zeta_{1}}{n}$, $\eta_{2} = -\left(\kappa_{1}\zeta_{2} + \frac{\kappa_{2}\zeta_{2}}{n^{2}}\right)$, $ \eta_{3} = \frac{\kappa_{2}\zeta_{1}}{n}$, $\eta_{4} = - \frac{\kappa_{2}\zeta_{2}}{n^{2}}$, and $\eta_{5} = - \frac{2\kappa_{2}\zeta_{2}}{n^{2}}$.

The payoff function is also given as $\Theta_{i}(w_{i}(CT),w_{-i}(CT))$, where $w_{-i}(CT)$ is the strategy vector of $\{\mathcal{V}^{N}_{j}\}$, $\forall j \in \mathcal{V}^{N}$, and $j \neq i$. The condition for a strategy vector $w^{*}(CT)$  to be the Nash equilibrium strategy $\forall \; \mathcal{V}^{N}_{i} \in \mathcal{V}^{N}$ applied in the period $[(C-1)T,CT]$ is given in~\eqref{nashcond}.
\begin{equation}\label{nashcond}
 \Theta_{i}(w^{*}_{i}(CT),w^{*}_{-i}(CT)) \geq \Theta_{i}(w_{i}(CT),w^{*}_{-i}(CT)) 
\end{equation}
We define the condition for Nash equilibrium for the payoff function through the result given in Lemma~\ref{lemma1}.

% For the component $U_{i,2}$, the variable is $\log_{2} \Upsilon_{i}(CT)$.
% \begin{equation}\label{util2}
%      U_{i,2} = \eta_{1}\log_{2} \Upsilon_{i}(CT) - \eta_{2}(\log_{2} \Upsilon_{i}(CT))^{2}
% \end{equation}
% The payoff for the game is given below.
% \begin{equation}\label{gameutil}
%      U = \sum_{i \in \mathcal{V}_{N}} U_{i} = \kappa_{1} \sum_{i \in \mathcal{V}_{N}} U_{i,1} + \kappa_{2} \sum_{i \in \mathcal{V}_{N}} U_{i,2}
% \end{equation}
\begin{lem} \label{lemma1}
  For the non-cooperative game defined in~\eqref{game} with a bounded strategy set and fixed number of users for a DAG-based DLT network, if $\frac{\zeta_{1}}{2\zeta_{2}} =  K_{cross}$, then the strategy $w_{i}(CT) =  K_{cross} $ is the Nash equilibrium $\forall \mathcal{V}^{N}_{i} \in \mathcal{V}^{N}$.
\end{lem}
\begin{proof}
    The detailed proof of the Lemma is given in Appendix~\ref{app:A}.
\end{proof}
Also, for the condition described in Lemma~\ref{lemma1}, the Nash equilibrium does not depend on the value $\kappa_{R}$, provided $\kappa_{2} \neq 0$. Detailed proof of the same is given in the form of corollary~\ref{cor1} in Appendix~\ref{app:A}.

From Lemma~\ref{lemma1}, we concur that the best scenario to utilize maximum network potential with minimum total computational resources across the system is when every node $\in \mathcal{V}^{N}$ issue $K_{cross}$ transactions in a time period $T$. If a node attempts to cross the given limit and cut another node's share of transactions, he gets penalized in the subsequent period. The penalty is a disproportionately higher difficulty level as per~\eqref{newdifflevel}. Therefore, it is in the best interest of every node to add transactions less than or equal to $K_{cross}$. Now, as we move forward, we discuss the case of the system with a dynamic set of nodes.
\subsection{Nash equilibrium for variable node set}
The equilibrium condition proved in Lemma~\ref{lemma1} is based on the fixed number of players in the game. However, for the given decentralized system, any new node can join at any point. We prove that the requirement of model-prescribed behavior does not change. As we proceed, we will demonstrate the consistency of the requirement of prescribed behavior through the proof of uniform shift of the Nash equilibrium. 

The non-cooperative game for shift in the average capacity $\Gamma_{(C-1)T}^{CT}$ in a variable node set is defined in~\eqref{vargame}, where $\mathcal{V}^{(N,(C-1)T)}$ is the effective set of players in the period $[(C-1)T,CT]$ with $|\mathcal{V}^{(N,(C-1)T)}| = n_{1}$, while $\mathcal{V}^{(N,CT)}$ is the effective user set for the period $[CT, (C+1)T]$ with $|\mathcal{V}^{(N,CT)}| = n_{2}$. The change from $n_{1}$ to $n_{2}$ occurs in the period $[(C-1)T, CT]$, but comes to effect from the period $[CT,(C+1)T]$. The system dynamics for the updated game are defined in~\eqref{vardimfunc}, \eqref{varprevnetcap} and~\eqref{varupnetcap}; in addition to the ones for the game defined in~\eqref{game} with the fixed set of players. The change in user set leads to the shift in average capacity in the period $[CT, (C+1)T]$. The difficulty level allotment is based on the node strategies with an updated set applied in the period $[(C+1)T, (C+2)T]$. Therefore, the Nash equilibrium gets updated for the period $[CT, (C+1)T]$.
\begin{equation}\label{vargame}
\begin{aligned}
    \Gamma_{(C-1)T}^{CT}  \triangleq & \{\mathcal{V}^{(N,(C-1)T)}, \mathcal{V}^{(N,CT)}, (w_{i}(CT))_{i \in \mathcal{V}^{(N,(C-1)T)}},\\ &(w_{i}((C+1)T))_{i \in \mathcal{V}^{(N,CT)}},\\
    &(\Theta^{(C-1)T}_{i})_{i \in \mathcal{V}^{(N,(C-1)T)}},(\Theta^{CT}_{i})_{i \in \mathcal{V}^{(N,CT)}} \}
    \end{aligned}
\end{equation}

For the given game, the set of players in successive periods, as well as the payoffs of successive periods, are the primary elements. We assess the shift in the Nash equilibrium through the result in Lemma~\ref{lemma2}.
\begin{lem}\label{lemma2}
The shift in the Nash equilibrium for the game defined in~\eqref{vargame}, when the number of users change from $n_{1}$ to $n_{2}$ during the period $[(C-1)T,CT]$ will be equal to $\left|\frac{\zeta_{1}(CT)}{2\zeta_{2}(CT)} - \frac{\zeta_{1}((C-1)T)}{2\zeta_{2}((C-1)T)}\right|$, where $\frac{\zeta_{1}(CT)}{2\zeta_{2}(CT)} = K_{cross}^{CT}$ and $\frac{\zeta_{1}((C-1)T)}{2\zeta_{2}((C-1)T)} = K_{cross}^{(C-1)T}$.
\end{lem}
\begin{proof}
    Detailed proof of the Lemma is given in Appendix~\ref{app:B}.
\end{proof}

From Lemma~\ref{lemma2}, we concur that in the instance of change in the user set, the shift in the Nash equilibrium is uniform across the system. No particular node gets an advantage in terms of preferred transaction limit due to an increase or decrease in the cardinality of the user set. Therefore, our proposed PoW difficulty level model ensures the continuity of the enforcement of prescribed behavior for a dynamic set of nodes. As we proceed, we explain how our model works through a set of solved examples.

\section{Numerical Examples}
\label{sec:example}
In this section, we demonstrate the efficacy of our model through numerical examples along with the analysis of the IOTA's model. Note that the network in question is decentralized, so there is no administrative authority to regulate the flow or restrict the offenders. 

\subsubsection*{\textbf{Example 1 (Analysis of IOTA congestion control model)}}
In the IOTA's model given in~\eqref{iotapow}, the range for
$\gamma_{i}$ is obtained through either of the two ways, normalization or relativity. Either way, this parameter assigns a node $\mathcal{V}^{N}_{i}$ to have the $\gamma_{i}$ in the range $[0,1]$. As high as the node spending is in the preceding time period, the lower its value will be, i.e., closer to $0$. 

For comparison, suppose we have two users $\mathcal{V}^{N}_{i}, \mathcal{V}^{N}_{j} \in \mathcal{V}^{N}$ as part of the network. Suppose the previous spending of $\mathcal{V}^{N}_{i}$ enables it to have $\gamma_{i} = 0.1$, while the same for $\mathcal{V}^{N}_{j}$ is $\gamma_{j} = 0.7$.
Now, if the number of transactions for $i$ and $j$ be $a_{i}(t) = a_{j}(t) = 10$, then we have the respective difficulty levels as $d_{i}(t) = d_{0} + 1$ and $d_{j}(t) = d_{0} + 7$.

The increment in $d_{i}(t)$ by a factor of $1$ leads to a doubling of the average computational efforts required. Due to the above results, $\mathcal{V}^{N}_{i}$ is able to put through its transactions almost instantly because of negligible change in assigned difficulty level. On the other hand, it becomes computationally very expensive for $\mathcal{V}^{N}_{j}$ to push its transactions instantaneously. In the above case, $\mathcal{V}^{N}_{j}$ has to wait for some time so that the value $d_{j}(t)$ comes down to a reasonable level for it. Therefore, the linearly increasing difficulty level provides respite to high stakeholders but discriminates against the users with low resources and spending power.

\subsubsection*{\textbf{Example 2}}
Suppose we have a network of $10$ users with the set given as $\mathcal{V}^{N} = \{\mathcal{V}^{N}_{1},\cdots, \mathcal{V}^{N}_{10}\}$. At present, we are proceeding with the assumption that neither a new node joins nor an existing node leaves. The time period for updating the reputation and level is $T = 10$ seconds. Suppose the network capacity is $100$ transactions per second, leading to the values $\Lambda = 1000$ transactions, $K_{cross} = 100$, and $K_{add} = 0$ for the period $T$. Also, let the current period be $[4T,5T]$. We form the transaction frequency-based reputation component $f^{w}_{i}(w_{i}(5T))$ for the node $\mathcal{V}^{N}_{i} \in \mathcal{V}^{N}$ by taking $K_{cross}$ with an invalid transaction value $-20$ shown in~\eqref{dimfuncex}. 
\begin{equation}\label{dimfuncex}
    f^{w}_{i}(w_{i}(5T)) = 1 + 0.04w_{i}(5T) - 0.0005w_{i}^{2}(5T) 
\end{equation}
The figure~\ref{repfig} demonstrates the variation of $f^{w}_{i}(w_{i}(5T))$ for different $\Lambda$ values for $10$ node system. We choose to proceed with the scenario of $\Lambda = 1000$ for the reputation model $ f^{w}_{i}(w_{i}(5T))$ in~\eqref{dimfuncex}.
\begin{figure}
\includegraphics[width=9.25cm]{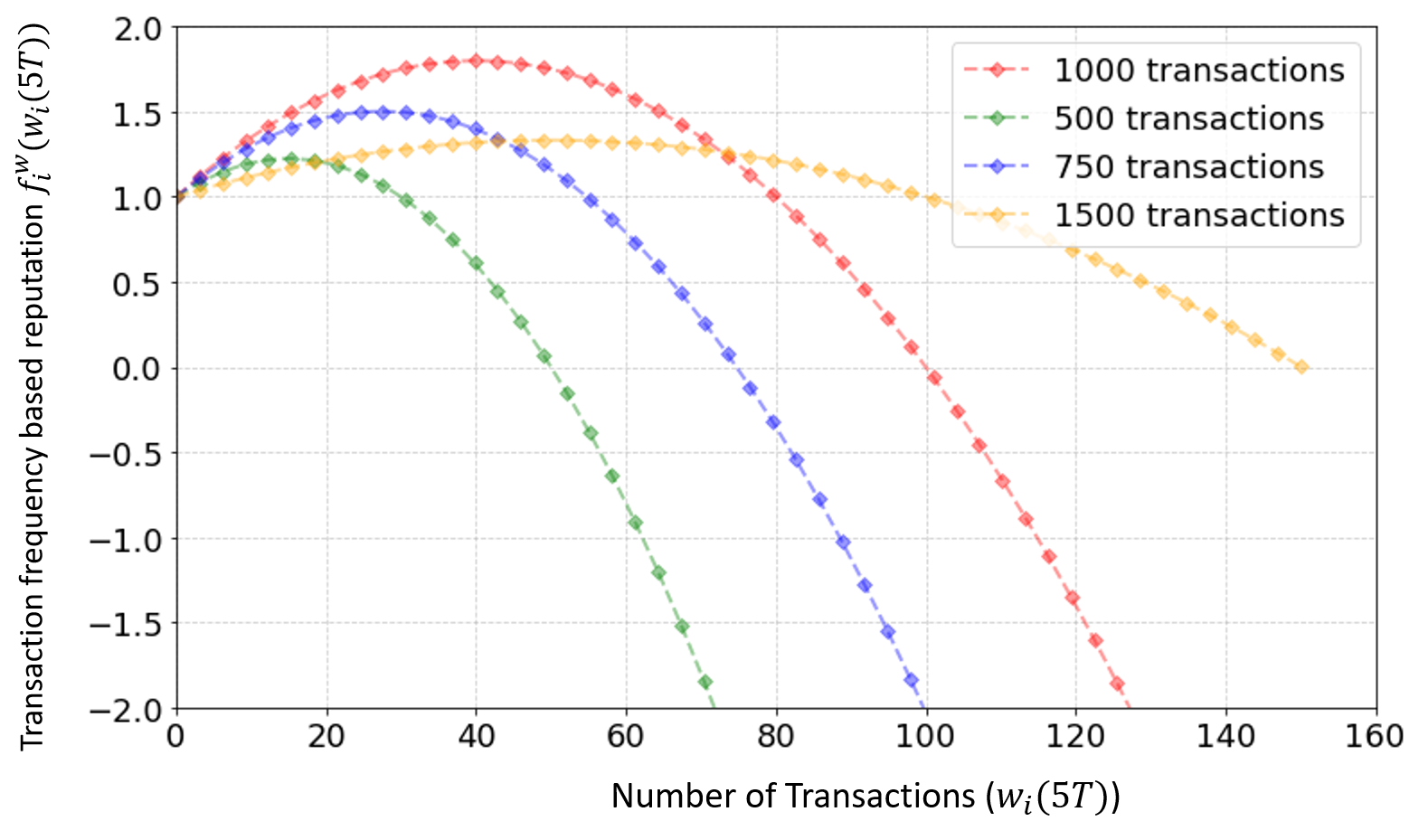}
\centering
\caption{Variation of transaction frequency based reputation}
\label{repfig}
\end{figure}
The value of $K_{tan} = 40$. For $0 < w_{i}(5T) \leq 100$, the difficulty level assigned for the period $[5T,6T]$ is represented in~\eqref{ex2eq1} by $d_{i}(5T) = d_{0} - \lfloor (K_{1} - 2K_{2}w_{i}(CT)) \rfloor$
\begin{equation}\label{ex2eq1}
    d_{i}(5T) = d_{0} - \lfloor (0.04 - 0.001w_{i}(5T)) \rfloor
\end{equation}
For $w_{i}(5T) \leq 40$, $d_{i}(5T) = d_{0}$. For $40 < w_{i}(5T) \leq 100$, $d_{i}(5T) = d_{0} + 1$.

For $w_{i}(5T) > 100$, the difficulty level is given in~\eqref{ex2eq2}.
\begin{equation}\label{ex2eq2}
    d_{i}(5T) = d_{0} + 1 + \lceil f^{B}_{i}(5T) \times (-1 - 0.04w_{i}(CT) + 0.0005w_{i}^{2}(CT)) \rceil
\end{equation}
For $w_{i}(5T) = 150$, the difficulty level is given in~\eqref{ex2eq3}.
\begin{equation}\label{ex2eq3}
    d_{i}(5T) = d_{0} + 1 + \lceil f^{B}_{i}(5T) \times 4.25 \rceil
\end{equation}
For $f^{B}_{i}(5T) = 0.2$, the above expression will be equal to $d_{0} + 2$. It means that $\mathcal{V}^{N}_{i}$ has to do $80 \%$ of token spending during the period $[4T,5T]$ for given level. Now, we have 10 users in the network, so the normal spending ratio is expected to be $10 \%$. For this ratio, we have $f^{B}_{i}(5T) = 0.9$ and subsequently $d_{i}(5T) = d_{0} + 1+ 4 = d_{0} + 5$.

For $w_{i}(5T) = 125$, the user $\mathcal{V}^{N}_{i}$ has to do $90 \%$ of the spending during the period $[4T,5T]$ to keep the difficulty level at $d_{i}(5T) = d_{0} + 2$. For $10 \%$ spending, $d_{i}(5T) = d_{0} + 3$. The analysis above showing the variation of difficulty level with a transaction-based reputation for different proportions of spending is plotted in figure~\ref{difflevfig} for the base level $d_{0} = 2$.  
\begin{figure}
\includegraphics[width=9.25cm]{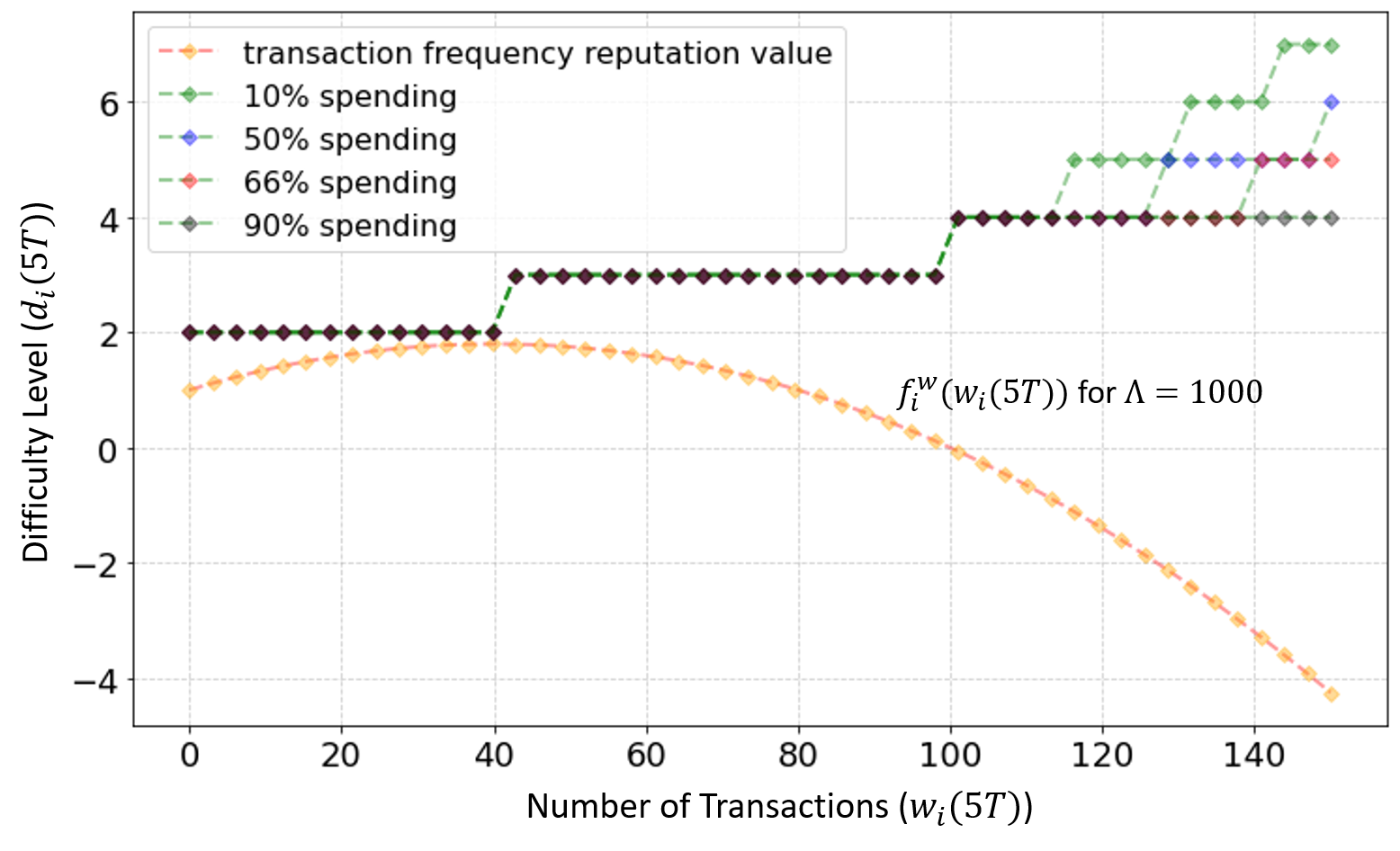}
\centering
\caption{Variation of difficulty level with transaction frequency}
\label{difflevfig}
\end{figure}
\subsubsection*{\textbf{Example 3}}To show the efficacy of prescribed behavior, we consider a system with 2 users $\mathcal{V}^{N}_{i}, \mathcal{V}^{N}_{j} \in \mathcal{V}^{N}$ in the network with network capacity $\Lambda = 210$, we have $K_{cross} = 100$ and $K_{add} = 10$. The time period for the difficulty level allotment is $[4T, 5T]$ based on transaction frequency in the period $[3T,4T]$. From Lemma~\ref{lemma1} and corollary~\ref{cor1}, we know that the values $\kappa_{1},\kappa_{2} \in \Re^{+}$ do not have effect on the outcome. So we let the share of outcome for the individual and average strategy be equal, i.e., $\kappa_{1} = \kappa_{2} = 0.5$. For $\frac{\zeta_{1}}{2\zeta_{2}} =  100$, let $\zeta_{1} = 100$ and $\zeta_{2} = 0.5$. The payoff function $\Theta_{i}$ for $\mathcal{V}^{N}_{i}$ with the calculated parameter values is given in~\eqref{payoffex1}. In figure~\ref{payoffmeshfig}, the variation of $\Theta_{i}$ for different values of $w_{i}(4T)$ and $w_{j}(4T)$ is provided on a $3D$ plane with peak value of $5000$. 
%Both figures conclude that the peak utility occurs when each node has transaction frequency $w_{i}(4T) = w_{j}(4T) = K_{cross}$.
\begin{equation}\label{payoffex1}
\begin{split}
    \Theta_{i} =75 w_{i}(4T) - 0.3125 w^{2}_{i}(4T) + 25  w_{j}(4T)\\
    - 0.0625  w^{2}_{j}(4T)  - 0.125 w_{i}(4T) w_{j}(4T) 
    \end{split}
\end{equation}
\begin{figure}
\includegraphics[width=9.25cm]{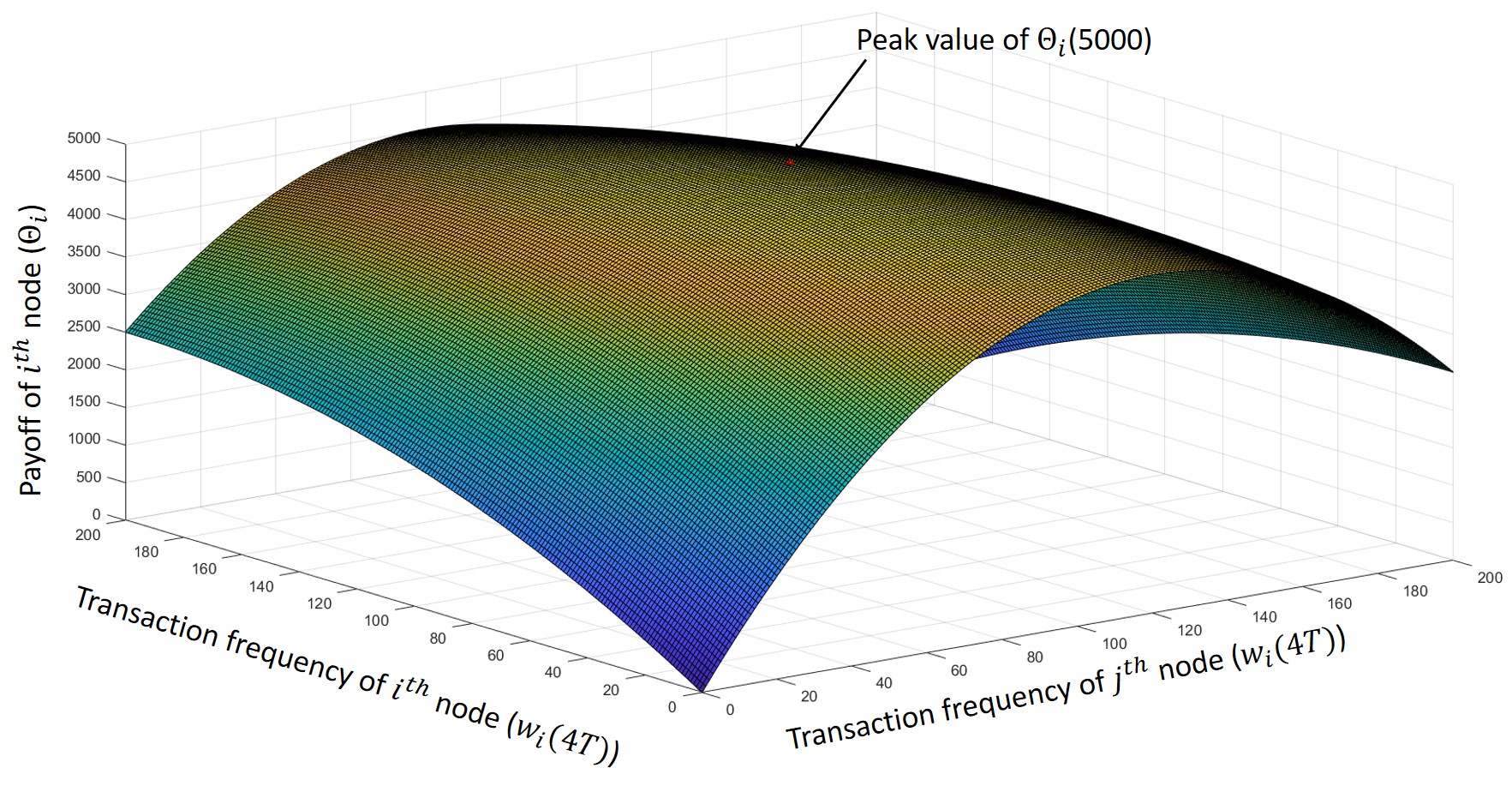}
\centering
\caption{Variation of $\mathcal{V}^{N}_{i}$ payoff ($\Theta_{i}$)}
\label{payoffmeshfig}
\end{figure}
A similar payoff function can also be computed for $\mathcal{V}^{N}_{j}$. From the payoff function, we derive the best response function for $\mathcal{V}^{N}_{i}$ as given in~\eqref{brfuncex3}. Similar payoff and best response functions can also be computed for $\mathcal{V}^{N}_{j}$. The existence of Nash equilibrium is demonstrated in figure~\ref{brfunc2userfig}, obtained at the point of intersection. 
\begin{equation}\label{brfuncex3}
b_{i}(w_{j}(4T)) =  120 -  0.2w_{j}(4T) 
\end{equation}
\begin{figure}
\includegraphics[width=9.25cm]{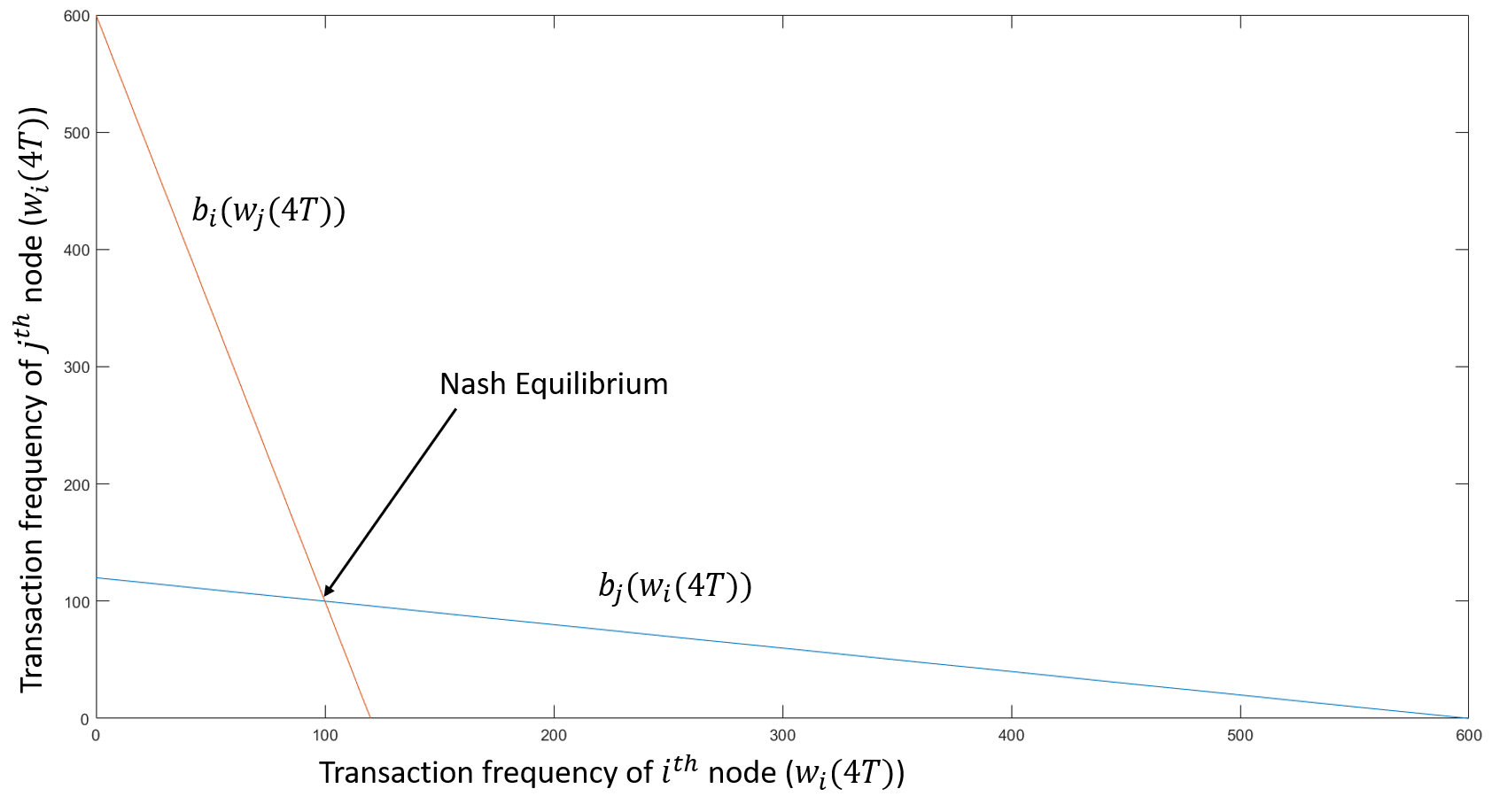}
\centering
\caption{Best response functions of $\mathcal{V}^{N}_{i}$ and $\mathcal{V}^{N}_{j}$}
\label{brfunc2userfig}
\end{figure}
From the best response functions, the Nash equilibrium is found to be $100$, i.e., equal to $K_{cross}$.

\subsubsection*{\textbf{Example 4}}
After showing the efficacy of prescribed behavior through the establishment of Nash equilibrium, we proceed to show that the change in the set of nodes does not lead to any deviation. For 2 users $\mathcal{V}^{N}_{i}, \mathcal{V}^{N}_{j} \in \mathcal{V}^{(N,4T)}$ in the network with network capacity $\Lambda^{4T} = 210$ and $n_{1} = 2$, we have $K^{4T}_{cross} = 100$ for the period $[3T,4T]$. Subsequently, $\frac{\zeta_{1}(4T)}{2\zeta_{2}(4T)} = 100$  can be taken as $\zeta_{1}(4T) = 100$ and $\zeta_{2}(4T) = 0.5$.
When a new user $\mathcal{V}^{N}_{k}$ joins the system between this period, the user set now becomes $\mathcal{V}^{N}_{i}, \mathcal{V}^{N}_{j}, \mathcal{V}^{N}_{k} \in \mathcal{V}^{(N,5T)}$ with $n_{2} = 3$ effective from the period $[4T,5T]$. Suppose the updated network capacity is $\Lambda^{5T} = 250$. The updated $K^{5T}_{cross}$ will be $80$. Subsequently, $\frac{\zeta_{1}(5T)}{2\zeta_{2}(5T)} = 80$  can be written as $\zeta_{1}(5T) = 80$ and $\zeta_{2}(5T) = 0.5$. For $\kappa_{1} = \kappa_{2} = 0.5$, the payoff function $\Theta^{4T}_{i}$ for $\mathcal{V}^{N}_{i}$ will be same as given in~\eqref{payoffex1}. 
%Substituting the values, we get the game parameter values given as $\eta_{1}(5T) = \frac{160}{3}$, $\eta_{2}(5T) = -\frac{2.5}{9}$, $\eta_{3}(5T) = \frac{40}{3}$, $\eta_{4}(5T) = - \frac{0.25}{9}$, and $\eta_{5}(5T) = - \frac{0.5}{9}$.

The updated payoff function $\Theta^{5T}_{i}$ for $n_{2} = 3$ is given in~\eqref{ex4payoff}. The distribution of the same is given in figure~\ref{payoff3userfig}. The distribution increases towards the higher end between the range of $70-90$. To find out the new Nash equilibrium with the shift, we derive and plot the updated best response functions for the set $\mathcal{V}^{(N,4T)}$ and $\mathcal{V}^{(N,5T)}$.
\begin{equation}\label{ex4payoff}
\begin{aligned}
   & \Theta^{5T}_{i} = \frac{160}{3} w_{i}(5T) -\frac{2.5}{9} w^{2}_{i}(5T)\\ &+ \frac{40}{3} (w_{j}(5T)+w_{k}(5T))
    - \frac{0.25}{9} (w^{2}_{j}(5T) +w^{2}_{k}(5T)) \\
   & - \frac{0.5}{9}( w_{i}(5T) w_{j}(5T) +w_{j}(5T) w_{k}(5T) 
   + w_{k}(5T) w_{i}(5T)) 
    \end{aligned}
\end{equation}
\begin{figure}
\includegraphics[width=9.25cm]{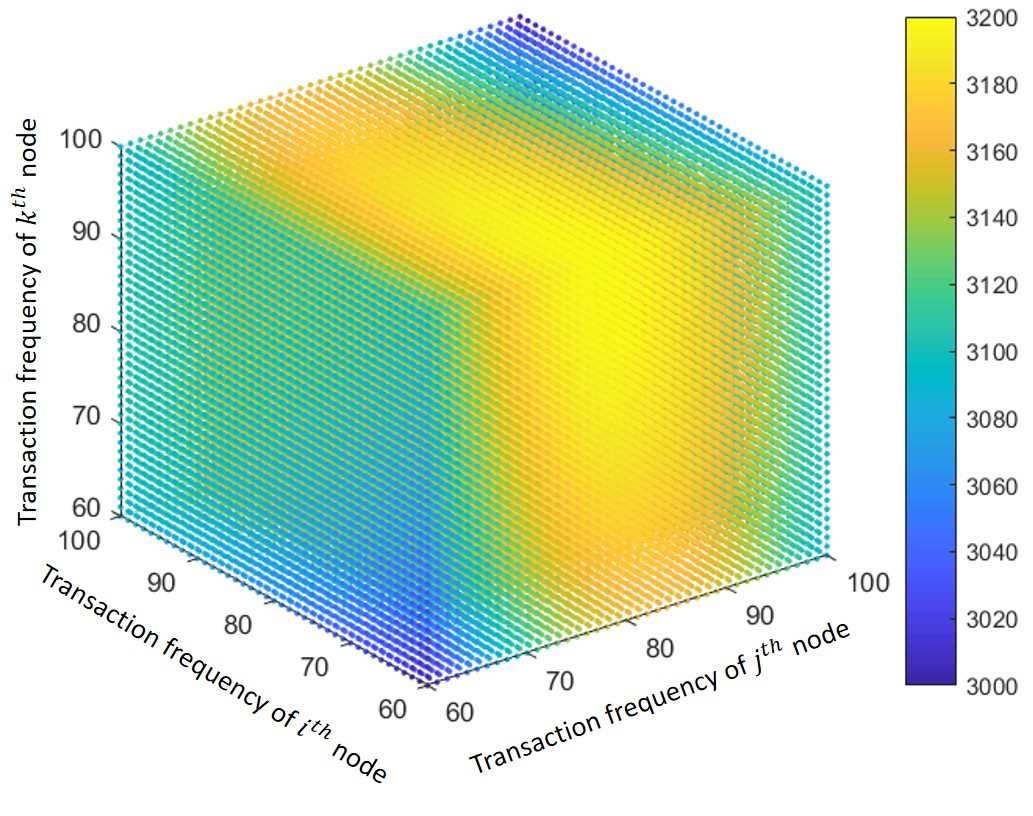}
\centering
\caption{Updated payoff of $\mathcal{V}^{N}_{i}$ for $n_{2} = 3$ node system}
\label{payoff3userfig}
\end{figure}
Now, we obtain the best response function for $\mathcal{V}^{N}_{i}\in \mathcal{V}^{(N,5T)}$ as shown in~\eqref{ex4brfunc}. By extension the same for $\mathcal{V}^{N}_{j}$ and $\mathcal{V}^{N}_{k}$ are $b^{5T}_{i}(w_{-i}(5T)) =   96 - 0.1(w_{j}(5T) + w_{k}(5T))$ and $b^{5T}_{k}(w_{-k}(5T)) =   96 - 0.1(w_{i}(5T) + w_{j}(5T))$ respectively. The plots for the best response functions with $n_{1} = 2$ and $n_{2} = 3$ are given in figure~\ref{ex4brfuncfig}. 
\begin{equation}\label{ex4brfunc}
b^{5T}_{i}(w_{-i}(5T)) =   96 - 0.1(w_{j}(5T) + w_{k}(5T))
\end{equation}
\begin{figure}
\includegraphics[width=9.25cm]{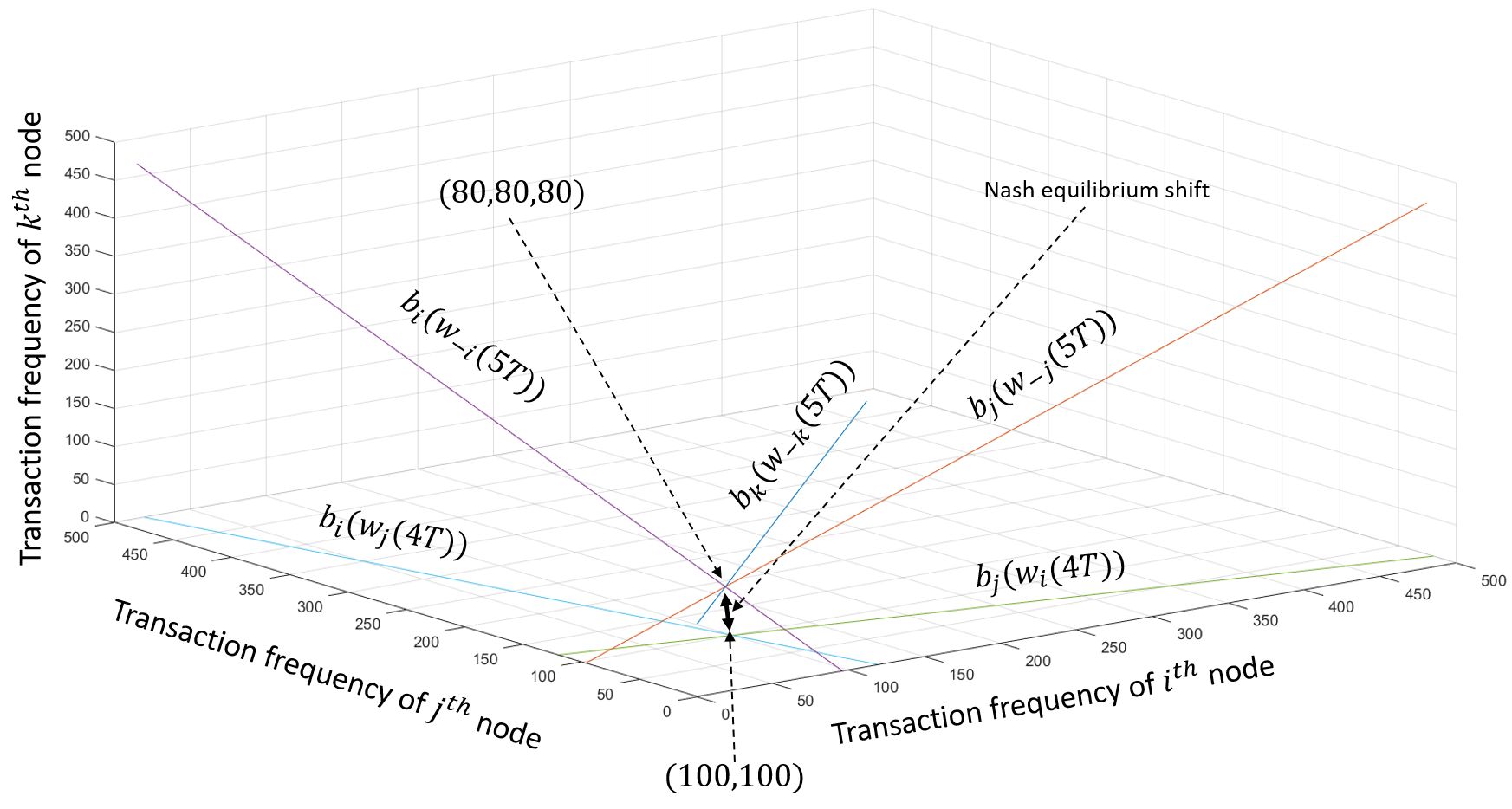}
\centering
\caption{Best response functions for $n_{1} = 2$ and $n_{2} = 3$ node systems}
\label{ex4brfuncfig}
\end{figure}
The unique solution for the set $\mathcal{V}^{(N,5T)}$ is $w_{i}(5T) = w_{i}(5T) = w_{k}(5T) = 80$, which the nash equilibrium for the given period. Therefore the shift in the nash equilibrium for the period $[4T,5T]$ from the period $[3T,4T]$ is $|100 -80| = 20$.

\section{Conclusion}
\label{sec:conc}
In this paper, we formulate a variable PoW model based on user behavior for a DAG-based distributed ledger network. It serves as the proof of attachment for a transaction into the DAG ledger. Our motive is to provide a solution to enable decentralized congestion control for DLT networks. For DLT networks, spamming of transactions by an unscrupulous individual or a group of users is the main cause of network congestion. While it is restricted in blockchain networks through different sets of measures, spamming is a potential problem in DAG-based DLT networks. Hence, we move forward with the PoW methodology for the DAG DLT.

The proposed method is a variable computational cost-based PoW difficulty level. We define the difficulty level for a time period based on the reputation built on transaction frequency and token spending in the previous period. It arranges for every user to have equal opportunities in terms of issuing transactions and establishes a prescribed transaction limit. For users violating the prescribed transaction limits, the model penalizes disproportionately by quadratic increase in the difficulty level value instead of linear. It prompts the nodes to stay within the prescribed limit when adding transactions into DLT. To show the efficacy of cooperative addition of transactions, we formulate the node behavior through a noncooperative game with a payoff based on the transaction frequency strategy. In the game, we establish the existence of a unique pure Nash equilibrium. When every node operates at Nash equilibrium, it utilizes the full potential of the available transaction limit with the minimum possible average computational resource. The model also works for a dynamic set of nodes, where the requirement of a uniform prescribed behavior remains the same. Overall, our proposed PoW difficulty level model is better for a DAG-based distributed ledger than one with a static or linearly increasing difficulty level.

% or


% \appendices
% \section{Proof of Lemma~\ref{lemma1}}
  
% \section{}

\section*{Acknowledgement}
This work is partially funded by the National Blockchain Project (grant number NCSC/CS/2017518) at IIT Kanpur, sponsored by the National Cyber Security Coordinator's office of the Government of India, the C3i Center funding from the Science and Engineering Research Board of the Government of India (grant number SERB/CS/2016466), and NSCS (National Security Council Secretariat).

\bibliographystyle{elsarticle-num}
\bibliography{Bibliography}

\appendix
\numberwithin{equation}{section}
\section{Values of $d_{i}(CT)$ for different values of $w_{i}(CT)$}\label{app:list}
\begin{itemize}
    \item The minimum value of $w_{i}(CT)$ at any instance is zero. At $w_{i}(CT) = 0$, $\frac{d f^{w}_{i}(w_{i}(CT))}{d w_{t}(CT)} = K_{1}$, while $f^{w}_{i}(w_{i}(CT)) = K_{0}$. Therefore from~eq\eqref{newdifflevel}, at $w_{i}(CT) = 0$, $d_{i}(CT) = d_{0} -\lfloor K_{1} \rfloor$ is the minimum value of the difficulty level.
    \item As the number of messages increases, the value of $d_{i}(CT)$ increases. At $w_{i}(CT) = \frac{K_{1}}{2K_{2}}$, $\frac{d f^{w}_{i}(w_{i}(CT))}{d w_{t}(CT)} = 0$. At this point, $d_{i}(CT) = d_{0}$.
    \item Afterwards, till $w_{i}(CT) <  K_{cross} $, $d_{i}(CT) = d_{0} + \lfloor 2K_{2}w_{i}(CT) - K_{1} \rfloor$.
    \item Let $\frac{K_{1}}{2K_{2}} = K_{tan}$. For $w_{i}(CT) \leq \lfloor K_{tan} \rfloor$, the difficulty level will be given as $d_{i}(CT) = d_{0} - \left\lfloor \frac{d f^{w}_{i}(w_{i}(CT))}{d w_{t}(CT)} \right\rfloor$ $= d_{0} - \lfloor (K_{1} - 2K_{2}w_{i}(CT)) \rfloor$, where $(K_{1} - 2K_{2}w_{i}(CT)) > 0$.
    \item At $w_{i}(CT) = \lfloor K_{tan} \rfloor$, the difficulty level will become equal to the base level, i.e., $d_{0}$, represented as $d_{i}(CT) = d_{0}$.
    \item For $\lfloor K_{tan} \rfloor < w_{i}(CT) \leq K_{cross}$, the differential term will cause the addition to the base level $d_{0}$ as $(K_{1} - 2K_{2}w_{i}(CT)) < 0$. However, the representation will remain same, i.e., $d_{i}(CT) = d_{0} - \left\lfloor \frac{d f^{w}_{i}(w_{i}(CT))}{d w_{t}(CT)} \right \rfloor$.
    \item For $w_{i}(CT) > K_{cross}$, where the transaction crosses the average suggested limit, the difficulty level equation changes to $d_{i}(CT) = d_{0} - \left\lfloor \frac{d f^{w}_{i}(w_{i}(CT))}{d w_{t}(CT)} \right\rfloor
    + \left\lceil f^{B}_{i}(CT) \times (- min(0,f^{w}_{i}(w_{i}(CT)))) \right\rceil$, also elaborated as equation~\eqref{eq:23}.
\end{itemize}
 
% \begin{equation}
%         d^{\lfloor K_{tan} \rfloor < w_{i}(CT) \leq  K_{cross} }_{i}(CT) = d_{0} - \left\lfloor \frac{d f_{w,i}(CT)}{d w_{t}(CT)} \right
%         \rfloor 
% \end{equation}
% For aforementioned case, $(K_{1} - 2K_{2}w_{i}(CT)) < 0$. The difficulty level equation is represented as equation~\ref{eq:22}.
% \begin{equation}\label{eq:22}
%         d^{\lfloor K_{tan} \rfloor < w_{i}(CT) \leq  K_{cross} }_{i}(CT) = d_{0} + \lfloor (2K_{2}w_{i}(CT) - K_{1}) \rfloor 
% \end{equation}

% \begin{equation}\label{eq:23}
%     \begin{split}
%     d^{w_{i}(CT) >  K_{cross}}_{i}(CT) = d_{0} - \left\lfloor \frac{d f_{w,i}(CT)}{d w_{t}(CT)} \right\rfloor
%     + \left\lceil f_{B,i}(CT) \times (- min(0,f_{w,i}(CT))) \right\rceil
% \end{split}
% \end{equation}
\begin{equation}\label{eq:23}
\begin{split}
    d^{w_{i}(CT) >  K_{cross} }_{i}(CT) = d_{0} - \lfloor (K_{1} - 2K_{2}w_{i}(CT)) \rfloor + \\
    \lceil f^{B}_{i}(CT) \times (-K_{0} - K_{1}w_{i}(CT) + K_{2}w_{i}^{2}(CT)) \rceil
\end{split}
\end{equation}

\section{Proof of Lemma~\ref{lemma1}}\label{app:A}
%\begin{lem} %\label{lemma1}
 \textbf{Lemma IV.1.} \textit{For the non-cooperative game defined in~\eqref{game} with a bounded strategy set and fixed number of users for a DAG-based DLT network, if $\frac{\zeta_{1}}{2\zeta_{2}} =  K_{cross}$, then the strategy $w_{i}(CT) =  K_{cross} $ is the Nash equilibrium $\forall \mathcal{V}^{N}_{i} \in \mathcal{V}^{N}$.}
%\end{lem}
\begin{proof}
A non-cooperative game has at least one pure strategy Nash equilibrium if its strategy set is compact and convex, and the payoff function is strictly concave and continuous in the strategy set for every user in the network~\cite{rosen65}. 

Now, the strategy set for our defined game is given as $w_{i}(CT) = [0, \Lambda]$. Since the given strategy set is closed and bounded, it is compact. 
The convexity of any set $S_{i}$ depends on the condition in~\eqref{cond} being satisfied, for any $ S^{1}_{i}, S^{2}_{i} \in S_{i}$ and for any $\theta \in [0,1]$.
\begin{equation}\label{cond}
    0 \leq \theta S^{1}_{i} + (1 - \theta) S^{2}_{i} \leq \Lambda
\end{equation}
The set $S_{i} \in \Re$ satisfies the above condition. Hence, it is convex $\forall i \in \mathcal{V^{N}}$. Using the payoff function $\Theta_{i}$, we derive the Hessian matrix for the game. Further, we check whether the Hessian matrix is negative definite $\forall s \in S$ or not. Here, $S$ is the strategy space for the whole network. The Hessian matrix $H_{s}$ is given in~\eqref{hessian}. For our case, $S_{i} = S$, i.e. each node has identical strategy space.

\begin{equation}\label{hessian}
    H_{s} = \begin{bmatrix}
\Theta^{''}_{11} & \Theta^{''}_{12}  & \cdots & \Theta^{''}_{1n} \\
 \Theta^{''}_{21} & \Theta^{''}_{22}  & \cdots & \Theta^{''}_{2n} \\
 \vdots& \vdots & \ddots & \vdots \\
 \Theta^{''}_{n1} & \Theta^{''}_{n2}  & \cdots & \Theta^{''}_{nn}
\end{bmatrix}
\end{equation}
Here $\Theta^{''}_{ij} = \frac{\partial^{2}\Theta_{i}}{\partial w_{i}(CT)\partial w_{j}(CT)}$, $\forall i,j \in \mathcal{V_{N}}$.
Now we compute the values of $\Theta^{''}_{ij}$ to check the conditions. For $\Theta^{''}_{ij}$, first we compute $\frac{\partial\Theta_{i}}{\partial w_{j}(CT)}$.

\begin{equation}\label{payoffderiv1}
    \frac{\partial\Theta_{i}}{\partial w_{j}(CT)} = \eta_{3}+ 2\eta_{4}w_{j}(CT) + \eta_{5} \sum^{i \neq j}_{i \in \mathcal{V}^{N}} w_{i}
\end{equation}

Now, to find $\Theta^{''}_{ij}$ for $i \neq j$, we partially differentiate~\eqref{payoffderiv1} by $w_{i}(CT)$ as given in~\eqref{doubderiv}. 

\begin{equation}\label{doubderiv}
    \Theta^{''}_{ij} = \frac{\partial^{2}\Theta_{i}}{\partial w_{i}(CT)\partial w_{j}(CT)} =  \eta_{5} = - 2\frac{\kappa_{2}\zeta_{2}}{n^{2}}
\end{equation}
Similarly, we find $\Theta^{''}_{ii}$ as given in~\eqref{doubderiv1}. 
\begin{equation}\label{doubderiv1}
    \Theta^{''}_{ii} = \frac{\partial^{2}\Theta_{i}}{\partial^{2} w_{i}(CT)} =  2 \eta_{2} = -(2\kappa_{1}\zeta_{2} + 2\frac{\kappa_{2}\zeta_{2}}{n^{2}})
\end{equation}
To verify whether the Hessian matrix $H(s)$ is negative definite or not, we check for its principal minors~\cite{mandy2018leading}. The condition for a matrix to be negative definite is that even order and odd order principal minors be positive and negative, respectively, based on Sylvester's criteria in the context of real and symmetric matrices~\cite{Hornmatrix},~\cite{chong2004introduction}. Based on values obtained for $ \Theta^{''}_{ii}$ and $ \Theta^{''}_{ij}$ in $H(s)$ $\forall \;\; i,j \in \mathcal{V}^{N}$, we safely conclude that the given matrix is negative definite. Now, if the Hessian matrix is negative definite, then the respective payoff function is strictly concave and continuous~\cite{rosen65}. Therefore, with establishing a compact and convex strategy set and the concave and continuous payoff function, the non-cooperative game defined in equation~\eqref{game} has at least one Nash equilibrium. 

Now, we find the uniqueness of Nash equilibrium using the best response function~\cite{fudenberg1991game}. We obtain the Nash equilibrium through proof by induction for the best response functions of each user. To start, let us have the case where we have 2 users, $\mathcal{V}^{N}_{i}$ and $\mathcal{V}^{N}_{j}$ in the network, i.e., $n = 2$. The utility~\eqref{derutil} for $\mathcal{V}^{N}_{i}$ with the substitution $\zeta_{1} = 2\zeta_{2}K_{cross}$ is represented in~\eqref{twoutil}.
\begin{equation}\label{twoutil}
\begin{aligned}
   & \Theta_{i} = (2\kappa_{1}\zeta_{2}K_{cross} + \kappa_{2}\zeta_{2}K_{cross}) w_{i}(CT) \\
   & -(\kappa_{1}\zeta_{2} + \frac{\kappa_{2}\zeta_{2}}{4}) w^{2}_{i}(CT) + \kappa_{2}\zeta_{2}K_{cross} w_{j}(CT)  \\
   & - \frac{\kappa_{2}\zeta_{2}}{4}  w^{2}_{j}(CT) - \frac{\kappa_{2}\zeta_{2}}{2} w_{i}(CT) w_{j}(CT) 
    \end{aligned}
\end{equation}
Let the ratio of the share of payoff be given as $\frac{\kappa_{1}}{\kappa_{2}} = \kappa_{R}$. For such a case, the payoff becomes as in~\eqref{redtwoutil}. 
\begin{equation}\label{redtwoutil}
\begin{aligned}
     & \Theta_{i} = \kappa_{2}\zeta_{2}((2\kappa_{R}K_{cross} + K_{cross}) w_{i}(CT) 
    -(\kappa_{R} + \frac{1}{4}) w^{2}_{i}(CT) \\
    &+ K_{cross} w_{j}(CT)  
    - \frac{1}{4}  w^{2}_{j}(CT) - \frac{1}{2} w_{i}(CT) w_{j}(CT)) 
    \end{aligned}
\end{equation}
To obtain the best response function, we compute the derivative of $\mathcal{V}^{N}_{i}$'s utility with respect to $w_{i}(CT)$ and equate it to $0$~\cite{owen2013game}. 
\begin{equation}\label{dertwoutil}
\begin{aligned}
   & \frac{\partial \Theta_{i}}{\partial w_{i}(CT)} = 2\kappa_{R}K_{cross} + K_{cross} \\
   &- (2\kappa_{R} + \frac{1}{2}) w_{i}(CT)  - \frac{1}{2} w_{j}(CT) = 0
\end{aligned}
\end{equation}
The best response function for $\mathcal{V}^{N}_{i}$ is given as below.
\begin{equation}\label{respwi}
 b_{i}(w_{j}(CT)) =  \frac{4\kappa_{R}K_{cross} + 2K_{cross}}{4\kappa_{R} + 1} - \frac{1}{4\kappa_{R} + 1}w_{j}(CT)
\end{equation}
Similarly, for $\mathcal{V}^{N}_{j}$, we obtain the best response function as~\eqref{respwj}. 
\begin{equation}\label{respwj}
b_{j}(w_{i}(CT)) =   \frac{4\kappa_{R}K_{cross} + 2K_{cross}}{4\kappa_{R} + 1} - \frac{1}{4\kappa_{R} + 1}w_{i}(CT)
\end{equation}
The Nash equilibrium for the best response equations \eqref{respwi} and \eqref{respwj} is the pair $\{w^{*}_{i}(CT), w^{*}_{j}(CT)\}$
such that $w^{*}_{j}(CT) = b_{j}(w^{*}_{i}(CT))$ and $w^{*}_{i}(CT) = b_{i}(w^{*}_{j}(CT))$~\cite{Varmanoncoop}. On solving, we obtain that the only solution for such a condition is $w^{*}_{i}(CT) = w^{*}_{j}(CT) = K_{cross}$. 

For generalized proof through induction, we assume that the Nash equilibrium is $K_{cross}$ for $n$ users in the system. Now, we proceed to obtain the Nash equilibrium for the system with $n+1$ users. For the $n$ users with $\mathcal{V}^{N} = \{\mathcal{V}^{N}_{1},\cdots,\mathcal{V}^{N}_{n}\}$, the payoff for $\mathcal{V}^{N}_{i}$ with the substitution $\zeta_{1} = 2\zeta_{2}K_{cross}$ and $\kappa_{1} = \kappa_{R}\kappa_{2}$ is given in~\eqref{rednuserutil}. 
\begin{equation}\label{rednuserutil}
\begin{aligned}
     & \Theta_{i} = \kappa_{2}\zeta_{2}((2\kappa_{R}K_{cross} 
     + \frac{2K_{cross}}{n}) w_{i}(CT) \\
     & -(\kappa_{R} + \frac{1}{n^{2}}) w^{2}_{i}(CT) 
    + \frac{2K_{cross}}{n}\sum^{j \neq i}_{j \in \mathcal{V}^{N}} w_{j}(CT)  \\
    & - \frac{1}{n^{2}} \sum^{j \neq i}_{j \in \mathcal{V}^{N}}  w^{2}_{j}(CT) - \frac{2}{n^{2}} \sum_{i,j \in [1,n]} w_{i}(CT) w_{j}(CT)) 
    \end{aligned}
\end{equation}
The best response function for $\mathcal{V}^{N}_{i}$ is obtained by partial differentiation as given in~\eqref{dernuserutil}.
\begin{equation}\label{dernuserutil}
\begin{aligned}
   \frac{\partial \Theta_{i}}{\partial w_{i}(CT)} =&  (2\kappa_{R}K_{cross} + \frac{2K_{cross}}{n}) \\
   &- 2(\kappa_{R} + \frac{1}{n^{2}}) w_{i}(CT) -  \frac{2}{n^{2}} \sum^{j \neq i}_{j \in [1,n]} w_{j}(CT)) = 0
\end{aligned}
\end{equation}
\begin{equation}\label{respwinuser}
\begin{aligned}
 b_{i}(w_{-i}(CT)) = & \frac{n^{2}\kappa_{R}K_{cross} + nK_{cross}}{n^{2}\kappa_{R} + 1} \\
 & -\frac{1}{n^{2}\kappa_{R} + 1}\sum^{j \neq i}_{j \in [1,n]} w_{j}(CT)
 \end{aligned}
\end{equation}
For the Nash equilibrium, the solution for the above is $b_{i}(w_{-i}(CT)) = w_{j}(CT) = K_{cross}$, $\forall j \in [1,n] - \{i\}$. Now, we consider the scenario for $n+1$ users with the same system conditions. The best response function for $w_{i}(CT)$ in such case is given in~\eqref{respwinplususer}.
\begin{equation}\label{respwinplususer}
\begin{aligned}
 & ((n+1)^{2}\kappa_{R} + 1)b_{i}(w_{-i}(CT)) =  (n+1)^{2}\kappa_{R}K_{cross} \\
 &+ (n+1)K_{cross} - \sum^{j \neq i}_{j \in [1,n+1]} w_{j}(CT)
\end{aligned}
\end{equation}
Expanding~\eqref{respwinplususer}, we get~\eqref{respwinplususer1}.
\begin{equation}\label{respwinplususer1}
\begin{aligned}
 & (n^{2}\kappa_{R} + 1 + 2n\kappa_{R} + \kappa_{R})b_{i}(w_{-i}(CT)) =  n^{2}\kappa_{R}K_{cross}\\ 
 & + nK_{cross} + \kappa_{R}K_{cross} +2n\kappa_{R}K_{cross} + K_{cross} \\
 &- \sum^{j \neq i}_{j \in [1,n]} w_{j}(CT) - w_{n+1}(CT)
\end{aligned}
\end{equation}
Using~\eqref{respwinuser}, we reduce~\eqref{respwinplususer1} to~\eqref{respwinplususer2}.
\begin{equation}\label{respwinplususer2}
b_{i}(w_{-i}(CT)) =  K_{cross} + \frac{K_{cross} - w_{n+1}(CT)}{2n\kappa_{R} + \kappa_{R}}
\end{equation}
Similarly, the best response function for $w_{n+1}(CT)$ in terms of $w_{i}(CT)$ is derived as in~\eqref{respwinplususer3}. 
\begin{equation}\label{respwinplususer3}
b_{n+1}(w_{-(n+1)}(CT)) =  K_{cross} + \frac{K_{cross} - w_{i}(CT)}{2n\kappa_{R} + \kappa_{R}}
\end{equation}
Based on the best response functions given in~\eqref{respwinplususer2} and~\eqref{respwinplususer3}, we form a similar set for all users in the system. For the aforementioned equations, the unique solution $\{w^{*}_{1},\cdots,w^{*}_{n+1}\}$, such that $w^{*}_{i} = b_{i}(w^{*}_{-i}(CT))$ $\forall i \in [1,n+1]$ is $K_{cross}$. Therefore, the value $K_{cross}$ is the Nash equilibrium for the given scenario.
\end{proof}

\begin{cor}\label{cor1}
For the system described in Lemma~\ref{lemma1}, the nash equilibrium does not depend on the value $\kappa_{R}$, provided $\kappa_{2} \neq 0$
\end{cor}
\begin{proof}
In~\eqref{respwinplususer2} and~\eqref{respwinplususer3}, we clearly see that $w_{i}(CT) = w_{n+1}(CT) = K_{cross}$ for the Nash equilibrium solution. The term $\kappa_{R}$ in the denominator has no effect on the outcome as the numerator is equal to $0$, while $n > 0$.
\end{proof}

\section{Proof of Lemma~\ref{lemma2}}\label{app:B}
%\begin{lem}%\label{lemma2}
\textbf{Lemma IV.2.} \textit{The shift in the Nash equilibrium for the game defined in~\eqref{vargame}, when the number of users change from $n_{1}$ to $n_{2}$ during the period $[(C-1)T,CT]$ will be equal to $\left|\frac{\zeta_{1}(CT)}{2\zeta_{2}(CT)} - \frac{\zeta_{1}((C-1)T)}{2\zeta_{2}((C-1)T)}\right|$, where $\frac{\zeta_{1}(CT)}{2\zeta_{2}(CT)} = K_{cross}^{CT}$ and $\frac{\zeta_{1}((C-1)T)}{2\zeta_{2}((C-1)T)} = K_{cross}^{(C-1)T}$.}
%\end{lem}
\begin{proof}
Through the proof of Lemma~\ref{lemma1}, we deduce that for two users $\mathcal{V}^{N}_{i}, \mathcal{V}^{N}_{j} \in \mathcal{V}^{(N,(C-1)T)}$, the Nash equilibrium is $K_{cross}^{(C-1)T}$ in the period $[(C-1)T,CT]$, with $\frac{\zeta_{1}((C-1)T)}{2\zeta_{2}((C-1)T)} = K_{cross}^{(C-1)T}$. The variable payoff function for the user set $\mathcal{V}^{(N,(C-1)T)}$ is given in~\eqref{varutil}. We proceed with the scenario of $n_{1} = 2$ and $n_{2} = 3$.
\begin{equation}\label{varutil}
    (\Theta^{(C-1)T}_{i})_{i \in \mathcal{V}^{(N,(C-1)T)}} = \kappa_{1}U^{(C-1)T}_{i} + \kappa_{2}U^{(C-1)T}_{av}
\end{equation}
For $n_{1} = 2$ in $\mathcal{V}^{(N,(C-1)T)}$, the derived utility function is written in~\eqref{varderutil}, where $\Theta_{i}((C-1)T)$ is the payoff obtained from the strategy $w_{i}(CT)$.
\begin{equation}\label{varderutil}
\begin{aligned}
    &(\Theta^{(C-1)T)}_{i})_{i \in \mathcal{V}^{(N,(C-1)T)}} = \eta_{1}((C-1)T) w_{i}((C-1)T) + \\
    & \eta_{2}((C-1)T) w^{2}_{i}((C-1)T) + \eta_{3}((C-1)T) \sum^{j \neq i}_{j \in \mathcal{V}_{N}} w_{j}((C-1)T) \\
    & + \eta_{4}((C-1)T) \sum^{j \neq i}_{j \in \mathcal{V}^{N}} w^{2}_{j}((C-1)T) \\
    &+ \eta_{5}((C-1)T)\sum_{i,j \in \mathcal{V}^{N}} w_{i}((C-1)T) w_{j}((C-1)T)
    \end{aligned}
\end{equation}
Now, suppose during $[(C-1)T,CT]$, a new node $\mathcal{V}^{N}_{k}$ joins and starts issuing transactions from the period $[CT, (C+1)T]$. The user set now becomes $\mathcal{V}^{(N,CT)} = \{\mathcal{V}^{N}_{i}, \mathcal{V}^{N}_{j}, \mathcal{V}^{N}_{k}\}$. In such case, the network capacity and subsequently average capacity changes, let it be $K^{CT}_{cross}$. The utility function also changes accordingly from $(\Theta^{(C-1)T}_{i})_{i \in \mathcal{V}^{(N,(C-1)T)}}$ to $(\Theta^{CT}_{i})_{i \in \mathcal{V}^{(N,CT)}}$. The game parameters are then written in the following form for the period $[CT,(C+1)T]$ with $n_{2}$ users, with the values derived based on constant values of static node set game. The values are $\eta_{1}(CT) = \kappa_{1}\zeta_{1}(CT) + \frac{\kappa_{2}\zeta_{1}(CT)}{n_{2}}$, $\eta_{2}(CT) = -(\kappa_{1}\zeta_{2}(CT) + \frac{\kappa_{2}\zeta_{2}(CT)}{n_{2}^{2}})$, $\eta_{3}(CT) = \frac{\kappa_{2}\zeta_{1}(CT)}{n_{2}}$, $\eta_{4}(CT) = - \frac{\kappa_{2}\zeta_{2}(CT)}{n_{2}^{2}}$, and $\eta_{5}(CT) = - \frac{2\kappa_{2}\zeta_{2}(CT)}{n_{2}^{2}}$.  

When the user set becomes $\mathcal{V}^{(N,CT)}$, the utility function $\Theta^{CT}_{i}$ for $\mathcal{V}^{N}_{i}$ for the period $[CT, (C+1)T]$ will be same as in~\eqref{varderutil} with updated constant values for $n_{2}$ nodes.
%With the addition of a new user into the set, the Nash equilibrium gets shifted due to the shift in payoff function parameters caused by a shift in the network capacity. 
To assess the shift in the Nash equilibrium, first, we take the scenario of two users $\mathcal{V}^{N}_{i}$ and $\mathcal{V}^{N}_{j}$ in the network during the period $[(C-1)T, CT]$. From corollary~\ref{cor1}, we know that the values of $\kappa_{R}$ have no effect. Therefore, we take $\kappa_{1} = \kappa_{2} = \kappa$ for the sake of simplicity. The best response function for $\mathcal{V}^{N}_{i}$ for the same period is derived as~\eqref{vartwoutil}.
\begin{equation}\label{vartwoutil}
\begin{split}
    \frac{\partial \Theta^{(C-1)T}_{i}}{\partial w_{i}((C-1)T)} = \kappa\zeta_{2}((C-1)T)(3K^{(C-1)T}_{cross} \\
    - \frac{5}{2} w_{i}((C-1)T)  - \frac{1}{2} w_{j}((C-1)T)) = 0
    \end{split}
\end{equation}
From the proof of Lemma~\ref{lemma1}, we ascertain the Nash equilibrium for the period $[(C-1)T, CT]$ to be $K^{(C-1)T}_{cross}$ for $n_{1} = 2$.

Now, we carry out the partial differentiation of $\Theta^{CT}_{i}$ with the user set $\mathcal{V}^{CT}_{N}$ as shown below.
\begin{equation}
\begin{split}
    \frac{\partial \Theta^{CT}_{i}}{\partial w_{i}(CT)} = \eta_{1}(CT) + 2\eta_{2}(CT) w_{i}(CT) \\
    + \eta_{5}(CT)( w_{j}(CT)  + w_{k}(CT))
    \end{split}
\end{equation}
Assuming the equal share of individual and average strategy payoff, we take $\kappa_{1} = \kappa_{2} = \kappa$. Substituting the value of weight parameters with the condition $\frac{\zeta_{1}(CT)}{2\zeta_{2}(CT)} = K^{CT}_{cross}$, we deduce the below expression. 
\begin{equation}
\begin{aligned}
  &\frac{\partial \Theta^{CT}_{i}}{\partial w_{i}(CT)} =   \kappa\zeta_{1}(CT) + \frac{\kappa\zeta_{1}(CT)}{3}  \\
   & -2(\kappa\zeta_{2}(CT) + \frac{\kappa\zeta_{2}(CT)}{9}) w_{i}(CT)  \\
   &- \frac{2\kappa\zeta_{2}(CT)}{9}( w_{j}(CT)  + w_{k}(CT))\\
   & = \kappa\zeta_{2}(CT)( \frac{8K^{CT}_{cross}}{3}  - \frac{20}{9} w_{i}(CT) - \frac{2}{9}( w_{j}(CT)  + w_{k}(CT)))  
   \end{aligned}
\end{equation}
Putting the partially differentiated equation equal to zero, we obtain the best response function of $\mathcal{V}^{N}_{i}$ with respect to strategies of $\mathcal{V}^{N}_{j}$ and $\mathcal{V}^{N}_{k}$ in~\eqref{brfuncs}. In a similar way, we can write the best response functions of $\mathcal{V}^{N}_{j}$ and $\mathcal{V}^{N}_{k}$ with respect to the remaining two users' strategies.
\begin{equation}\label{brfuncs}
b^{CT}_{i}(w_{-i}(CT)) =   1.2 K^{CT}_{cross} - 0.1(w_{j}(CT) + w_{k}(CT))  
\end{equation}
Solving the best response functions, the strategy $w_{i}(CT) = w_{j}(CT) = w_{k}(CT) = K^{CT}_{cross}$ is found to be the new Nash equilibrium for the network. Therefore, the shift in the Nash equilibrium, i.e., $|K^{CT}_{cross} - K^{(C-1)T}_{cross}|$ will be equal to $|\frac{\zeta_{1}(CT)}{2\zeta_{2}(CT)} - \frac{\zeta_{1}((C-1)T)}{2\zeta_{2}((C-1)T)}|$. If $n_{1} = 3$ and $n_{2} = 2$, i.e., a user leaves the system, the shift in the Nash equilibrium will be $|K^{CT}_{cross} - K^{(C-1)T}_{cross}|$ in a similar way.

If we talk in generalized terms, the best response function for $\mathcal{V}^{N}_{i}$ is derived based on the observations from the lemma~\ref{lemma1}. For $\mathcal{V}^{N}_{i} \in \mathcal{V}^{(N,(C-1)T)}$ with $|\mathcal{V}^{(N,(C-1)T)}| = n_{1}$ for the period $[(C-1)T, CT]$, the best response function is given in~\eqref{brfuncn1}. 
\begin{equation}\label{brfuncn1}
\begin{aligned}
   & b^{(C-1)T}_{i}(w_{-i}((C-1)T)) =  \frac{n^{2}_{1}K^{(C-1)T}_{cross} + n_{1}K^{(C-1)T}_{cross}}{n^{2}_{1} + 1} \\
   & - \frac{1}{n_{1}^{2} + 1}(\sum^{j \neq i}_{j \in \mathcal{V}^{(N,(C-1)T)}} w_{j}((C-1)T))
    \end{aligned}
\end{equation}
We obtain the unique solution from the above set of equations, which is $\forall \mathcal{V}^{N}_{i} \in \mathcal{V}^{(N,(C-1)T)}$, $w_{i}((C-1)T) = K^{(C-1)T}_{cross}$. 

Similarly, for $\mathcal{V}^{N}_{i} \in \mathcal{V}^{(N,CT)}$ with $|\mathcal{V}^{(N,CT)}| = n_{2}$ for the period $[CT, (C+1)T]$, the best response function based on Lemma~\ref{lemma1} is given in~\eqref{brfuncn2}.  
\begin{equation}\label{brfuncn2}
\begin{aligned}
b^{CT}_{i}(w_{-i}(CT)) =  &  \frac{n^{2}_{2}K^{CT}_{cross} + n_{2}K^{CT}_{cross}}{n^{2}_{2} + 1} \\
&- \frac{1}{n_{2}^{2} + 1}(\sum^{j \neq i}_{j \in \mathcal{V}^{(N,CT)}} w_{j}(CT))
\end{aligned}
\end{equation}
In a similar way as in the case of $n_{1}$ users, we obtain the unique solution for the set of $n_{2}$ users, which is $\forall \mathcal{V}^{N}_{i} \in \mathcal{V}^{(N,CT)}$, $w_{i}(CT) = K^{CT}_{cross}$.

From the set~\eqref{brfuncn1} and~\eqref{brfuncn2}, the Nash equilibrium points for the user sets $\mathcal{V}^{(N,(C-1)T)}$ and $\mathcal{V}^{(N,CT)}$ are $K^{(C-1)T}_{cross}$ and $K^{CT}_{cross}$ respectively. Based on the assumptions related to the parameters, the shift in the Nash equilibrium, i.e., $|K^{CT}_{cross} - K^{(C-1)T}_{cross}|$ will be equal to $\left|\frac{\zeta_{1}(CT)}{2\zeta_{2}(CT)} - \frac{\zeta_{1}((C-1)T)}{2\zeta_{2}((C-1)T)}\right|$.  

\end{proof}


% Author biography without author photo.
% Author biography. Author biography. Author biography.
% Author biography. Author biography. Author biography.
% Author biography. Author biography. Author biography.
% Author biography. Author biography. Author biography.
% Author biography. Author biography. Author biography.
% Author biography. Author biography. Author biography.
% Author biography. Author biography. Author biography.
% Author biography. Author biography. Author biography.
% Author biography. Author biography. Author biography.

% \bio{figs/cas-pic1}
% Author biography with author photo.
% Author biography. Author biography. Author biography.
% Author biography. Author biography. Author biography.
% Author biography. Author biography. Author biography.
% Author biography. Author biography. Author biography.
% Author biography. Author biography. Author biography.
% Author biography. Author biography. Author biography.
% Author biography. Author biography. Author biography.
% Author biography. Author biography. Author biography.
% Author biography. Author biography. Author biography.

% \bio{figs/cas-pic1}
% Author biography with author photo.
% Author biography. Author biography. Author biography.
% Author biography. Author biography. Author biography.
% Author biography. Author biography. Author biography.
% Author biography. Author biography. Author biography.

\end{document}